\documentclass[letter,10pt]{article}
\usepackage{fullpage}
\pdfoutput=1

\usepackage[cmex10]{amsmath}
\usepackage{amsthm}
\usepackage{amsxtra}
\usepackage{amssymb}
\usepackage{textcomp}
\usepackage{wasysym} 
\usepackage{graphicx} 
\usepackage{color} 
\usepackage[noadjust]{cite}
\usepackage{subfigure}
\usepackage{mathrsfs} 
\usepackage{centernot}
\usepackage{enumerate}
\usepackage{accents} 
\usepackage{ifpdf}

\usepackage{mathtools} 

\usepackage{lpic}

\usepackage{marginnote} 
\usepackage{color}  

\newtheorem{thm}{Theorem}
\newtheorem{cor}{Corollary}

\newtheorem{prop}{Proposition}

\newtheorem{rem}{Remark}

\newtheorem{defn}{Definition}

\DeclareMathOperator{\dist}{dist}

\DeclareMathOperator{\Reach}{Reach}
\DeclareMathOperator{\Disc}{Disc}
\DeclareMathOperator{\Viab}{Viab}

\DeclareMathOperator{\Dom}{Dom}
\DeclareMathOperator{\Init}{Init}

\newcommand{\inter}[1]{\accentset{\circ}{#1}}

\DeclareMathOperator{\cl}{cl}
\newcommand{\Y}{\mathcal{Y}}
\newcommand{\V}{\mathcal{V}}
\newcommand{\Vl}{\mathscr{V}}

\newcommand{\U}{\mathcal{U}}
\newcommand{\Ul}{\mathscr{U}}
\newcommand{\Z}{\mathcal{Z}}
\newcommand{\R}{\mathcal{R}}
\newcommand{\X}{\mathcal{X}}
\newcommand{\K}{\mathcal{K}}
\newcommand{\E}{\mathcal{E}}
\newcommand{\A}{\mathcal{A}}
\newcommand{\B}{\mathcal{B}}
\newcommand{\C}{\mathcal{C}}

\newcommand{\I}{\mathcal{I}}

\newcommand{\M}{\mathcal{M}}

\newcommand{\Q}{\mathcal{Q}}

\newcommand{\tr}[1]{#1^\textsl{T}}
\newcommand{\Real}{\mathbb{R}}

\newcommand{\abs}[1]{\lvert#1\rvert}

\newcommand{\norm}[1]{\left\lVert#1\right\rVert}
\newcommand{\inner}[1]{\langle#1\rangle}
\newcommand{\Inner}[1]{\big\langle#1\big\rangle}

\newcommand{\p}{\phantom{+}}

\newcommand{\ellpfunc}{\Circle}

\newcommand{\NA}{{\textup{na}}}
\newcommand{\CL}{{\textup{fb}}}
\newcommand{\InpIn}{\I_\text{i}}
\newcommand{\InpEx}{\I_\text{e}}

\begin{document}

\title{
Scalable Safety-Preserving Robust Control Synthesis for Continuous-Time Linear Systems\thanks{Research supported by NSERC Discovery Grants \#327387 (Oishi) and \#298211 (Mitchell), NSERC Collaborative Health Research Project \#CHRPJ-350866-08 (Dumont), and the Institute for Computing, Information and Cognitive Systems (ICICS) at UBC. This work was mainly carried out at Electrical \& Computer Engineering, University of British Columbia, Vancouver, BC, Canada.}}

\author{Shahab Kaynama\thanks{S.\ Kaynama ({\tt\small kaynama@eecs.berkeley.edu}, cor.\ author) is with Electrical Engineering \& Computer Sciences, University of California at Berkeley, 337 Cory Hall, Berkeley, CA 94720, USA.} , Ian M.\ Mitchell\thanks{I.\ Mitchell ({\tt\small mitchell@cs.ubc.ca}) is with Computer Science, University of British Columbia, 2366 Main Mall, Vancouver, BC V6T1Z4, Canada.} , Meeko Oishi\thanks{M.\ Oishi ({\tt\small oishi@unm.edu}) is with Electrical \& Computer Engineering, University of New Mexico, MSC01 1100, 1 University of New Mexico, Albuquerque, NM 87131, USA.} , Guy A.\ Dumont\thanks{G.\ Dumont ({\tt\small guyd@ece.ubc.ca}) is with Electrical \& Computer Engineering, University of British Columbia, 2332 Main Mall, Vancouver, BC V6T1Z4, Canada.}}

\date{(Preprint Submitted for Publication)}

\maketitle

\begin{abstract}
    We present a \emph{scalable} set-valued safety-preserving controller for constrained continuous-time linear time-invariant (LTI) systems subject to additive, unknown but bounded disturbance or uncertainty. The approach relies upon a conservative approximation of the discriminating kernel using robust maximal reachable sets---an extension of our earlier work on computation of the viability kernel for high-dimensional systems. Based on ellipsoidal techniques for reachability, a piecewise ellipsoidal algorithm with polynomial complexity is described that under-approximates the discriminating kernel under LTI dynamics. This precomputed piecewise ellipsoidal set is then used online to synthesize a permissive state-feedback safety-preserving controller. The controller is modeled as a hybrid automaton and can be formulated such that under certain conditions the resulting control signal is continuous across its transitions. We show the performance of the controller on a twelve-dimensional flight envelope protection problem for a quadrotor with actuation saturation and unknown wind disturbances.
\end{abstract}

\section{Introduction} \label{S:Intro}


Hard constraints on inputs and states appear in many cyberphysical systems due to safety requirements and actuator limitations. The ability of a system to operate successfully under such constraints becomes crucially important in the presence of adversarial disturbances and uncertainties. Therefore, designing permissive feedback controllers that can guarantee the system's satisfaction of all constraints whilst fulfilling given performance criteria is highly desirable.

The subset of state space for which safety-preserving controllers exist is referred to as the \emph{discriminating kernel}, or in the absence of a disturbance input, the \emph{viability kernel} \cite{aubin2011,cardaliaguet1999set}.\footnote{When constraint satisfaction is considered over an infinite time horizon, the viability (discriminating) kernel is equivalent to the \emph{largest (robust) controlled-invariant set} \cite{blanchini2008set}.} The control policies associated with these kernels are capable of keeping the trajectories of the system within the safe region of the state space. Their synthesis has therefore received significant attention among researchers \cite{aubin2011,blanchini2008set,Asarin_Bournez_Dang_Maler_Pnueli_2000}. Application domains in which safety must be maintained despite bounded inputs include control of depth of anesthesia \cite{kaynama_thesis}, aircraft envelope protection \cite{Margellos_Lygeros_2009,Lygeros1999,Tomlin2000a,TMBO03}, autolanders \cite{BMOT07}, collision avoidance \cite{MBT05,Del_Vecchio_Malisoff_Verma_2009,Hafner_Cunningham_Caminiti_Del_Vecchio_2011, Broucke_2006,Oishi_2001}, automated highway systems \cite{Lygeros_Godbole_Sastry_1998}, control of under-actuated underwater vehicles \cite{Panagou_2013}, stockout prevention of constrained storage systems in manufacturing processes \cite{Borrelli_Del_Vecchio_Parisio_2009}, management of a marine renewable resource \cite{Bene2001}, and safety in semi-autonomous vehicles \cite{Borrelli2012ACC}.


A classical approach to synthesizing safety-preserving controllers is to use the information of the shape of the computed kernel; as per Nagumo's theorem and its generalizations \cite{aubin1991viability}, the optimal control laws are chosen based on the contingent cone \cite{Bouligand32} or the proximal normal \cite{aubin1991viability,cardaliaguet1999set} at every point on the boundary of the set. Similarly, a given viability/discriminating kernel can be used as the terminal constraint set to guarantee recursive feasibility of the receding horizon optimization problem in a model predictive control (MPC) framework \cite{Blanchini_1999, Mayne_Rawlings_Rao_Scokaert_2000,Kerrigan_2000,Bemporad_Borrelli_Morari_2002, Rakovic_Mayne_2005,Rakovic_2009}. This results in constraint satisfaction by the closed-loop system and therefore the generated control laws are safety-preserving. In both cases discussed above the computational complexity of synthesizing such control laws is at best equal to that of computing the corresponding kernel.

Eulerian methods \cite{cardaliaguet1999set, MBT05, gao2006-reachability} are capable of computing the discriminating kernel and their associated safety-preserving control laws. However, these methods rely on gridding the state space and therefore suffer from a computational complexity that is exponential in the dimension of the state. Although versatile in their ability to handle various types of dynamics and constraints, the applicability of Eulerian methods is limited to systems of low dimensionality---usually less than four. Within the MPC framework, in the case of discrete-time LTI systems with polytopic constraints, the discriminating/viability kernel (which is also a polytope) is computed using algorithms such as \cite{KGBM04}. However, the complexity of this computation is also exponential and thus restricted to low-dimensional systems.\footnote{The number of vertices and facets of the polytope increases rapidly with each successive Minkowski sum and intersection operation. Additionally, the various vertex-facet enumerations involved are known to be NP-hard.} For the more general case of continuous-time systems, a common practice is to crudely compute a single ellipsoid (usually a control Lyapunov function \cite{blanchini2008set}) that is a subset of the discriminating kernel and use this ellipsoid as the terminal constraint set. If such a set can be found (and this in itself may be a big assumption), then with a linear cost functional the receding horizon optimization problem becomes a quadratic programme. However, the performance of such controller can be extremely conservative.

In this paper, we will first extend our earlier results \cite{kaynama_HSCC2012} on scalable computation of the viability kernel to the case of systems with adversarial inputs. We will begin by representing the discriminating kernel of a generic nonlinear system in terms of robust maximal reachable sets. Using this representation, we will describe a piecewise ellipsoidal algorithm for LTI systems that can be employed to compute a conservative approximation (i.e.\ a guaranteed under-approximation) of the discriminating kernel. Then, we will use the computed approximation to set up a \emph{scalable} robust safety-preserving controller for LTI systems based on the optimal control strategies described in \cite{kurzhanskii96ellipsoidal}. We formulate the controller as a hybrid automaton with two modes: safety and performance. The domains of the automaton are time-varying piecewise ellipsoidal sets that are precomputed offline during the analysis phase (i.e.\ when under-approximating the discriminating kernel). When online, in the performance mode, any admissible control input can be applied. For example, one or multiple performance objectives can be optimized (or even a model-free reinforcement learning controller can be used) without the need to care for safety or feasibility. The automaton transitions to the safety mode on the boundary of the piecewise ellipsoidal sets when safety is at stake. In this mode, if specific safety-preserving values of the control input are applied the safety of the closed-loop system is maintained. Since the most demanding aspect of this approach is the offline piecewise ellipsoidal approximation of the discriminating kernel, the complexity of the technique is polynomial (roughly cubic) in the dimension of the state. The online computations are of quadratic complexity but are inexpensive to perform (mainly matrix multiplications) if the chosen performance controller is of equal or less complexity.

The paper is organized as follows: In Section~\ref{S:prob_formulation} we formulate the problem and lay out necessary preliminaries. Section~\ref{S:approximation_control_phil} describes the approximation algorithm and the main ideas behind our scalable safety-preserving control strategy. These results are expressed in terms of a general nonlinear system. Section~\ref{S:analysis_algorithm} describes the offline phase which entails a piecewise ellipsoidal approximation of the discriminating kernel for LTI systems. Section~\ref{S:synthesis_algorithm} employs these precomputed sets to formulate the safety-preserving feedback controller. Discussions on infinite-horizon control and handling of chattering and Zeno executions are also provided. A twelve-dimensional example is considered in Section~\ref{S:Application}, and conclusions and future work are discussed in Section~\ref{S:conc}.

\subsection{Other Related Work}

Synthesis of safety-preserving controllers can be viewed as a search for an appropriate control Lyapunov function (albeit, subject to additional constraints and possibly a stabilizability assumption) and its associated control laws. To this end, an alternative to the aforementioned classical approach is related to techniques based on sum-of-squares (SOS) optimization for polynomial systems with semi-algebraic constraints; cf.\ ~\cite{tedrake2010lqr, tan_CLF_04, Rantzer_Prajna_2004}.\footnote{The synthesis technique in ~\cite{Rantzer_Prajna_2004} relies on state space gridding which renders it intractable in high dimensions.} The resulting bilinear program is solved either by alternating search or via convex relaxations. The degree of the SOS multipliers is commonly kept low (e.g.\ quartic), striking a tradeoff between excessive conservatism and computational complexity. A related technique is the recent method of occupation measures~\cite{henrion2012convex, majumdar2013technical} that, while convex and thus scalable, can only compute an over-approximation of the desired set (insufficient for safety). Thus, the resulting controllers are not guaranteed to be safety-preserving.


A classification technique based on support vector machines (SVMs) is presented in \cite{Deffuant_Chapel_Martin_2007} that approximates the viability kernel and yields an analytical expression of its boundary. A sequential minimal optimization algorithm is solved to compute the SVM that in turn forms a barrier function on or close to the boundary of the viability kernel in the discretized state space. While the method successfully reduces the computational time for the synthesis of control laws when the dimension of the input space is high, its applicability to systems with high-dimensional \emph{state space} is limited. The method does not provide any guarantees that the synthesized control laws are safety-preserving.




The notion of approximate bisimulation \cite{GP07_automatica} can be used to construct a discrete abstraction of the continuous state space such that the observed behavior of the corresponding abstract system is ``close'' to that of the original continuous system. Girard \emph{et al.} in a series of papers \cite{Girard_2010,Camara_Girard_Gossler_2011,Girard_2012} use this notion to construct safety-preserving controllers for approximately bisimilar discrete abstractions and prove that the controller is \emph{correct-by-design} for the original systems. The technique is applied to incrementally stable switched systems (for which approximately bisimilar discrete abstractions of arbitrary precision can be constructed) with autonomous or affine dynamics, and safety-preserving switched controllers are synthesized. The abstraction, however, relies on sampling of time and space (i.e.\ gridding) and therefore its applicability appears to be limited to low-dimensional systems---even when multi-scale abstraction (adaptive gridding) techniques are employed.

\subsection{Summary of Contributions}

We describe a novel iterative technique that expresses the discriminating kernel under general nonlinear dynamics in terms of robust maximal reachable sets. By doing so, we enable the use of efficient and scalable Lagrangian algorithms for maximal reachability for the computation of the discriminating kernel---a computation that has historically been a bottleneck for safety analysis and synthesis of high-dimensional systems.

Focusing on LTI dynamics, we then make use of one such Lagrangian algorithm, the ellipsoidal techniques for maximal reachability \cite{Kurzhanski2000a}, to 1) conservatively/correctly approximate the discriminating kernel over a given time horizon and do so in a scalable fashion, and 2) formulate a novel hybrid controller (based on the synthesis of the maximal reachability control laws) that maintains safety over \emph{at least} the given horizon by introducing the notion of pseudo-time as a continuous variable in the hybrid model. We present a sufficient condition and a variant of the hybrid controller that enable \emph{infinite-horizon} safety-preserving control based on finite-horizon analysis.  Moreover, we describe a technique that ensures that the resulting permissive safety-preserving control law is continuous across the transitions of the hybrid controller (which in turn can prevent chattering, a common problem in hybrid control of continuous-time systems).

\section{Problem Formulation} \label{S:prob_formulation}

Consider the continuous-time system
\begin{equation}\label{E:nonlinear_ss_eqn}
    \dot{x} = f(x,u,v), \quad x(0)=x_0
\end{equation}
with state space $\X:=\Real^n$ (a finite-dimensional vector space) and tangent space $T\X$, state vector $x(t)\in\X$, control input $u(t)\in\U$, and disturbance input $v(t)\in\V$, where $\U$ and $\V$ are, respectively, compact convex subsets of $\Real^{m_u}$ and $\Real^{m_v}$. The vector field $f\colon \X \times \U \times \V \to T\X$ is assumed to be Lipschitz in $x$ and continuous in both $u$ and $v$. Let $\Ul_{[0,t]}:=\{u\colon [0,t] \to \Real^{m_u} \;\text{measurable}, \; u(s)\in \U \; \text{a.e.} \; s\in [0,t] \}$ and $\Vl_{[0,t]}:=\{v\colon [0,t] \to \Real^{m_v} \;\text{measurable}, \; v(s)\in \V \; \text{a.e.}\;  s\in [0,t] \}$ be the sets of input signals. With an arbitrary, finite time horizon $\tau>0$, for every $t\in [0,\tau]$, $x_0\in \X$, $u(\cdot)\in \Ul_{[0,t]}$, and $v(\cdot)\in \Vl_{[0,t]}$, there exists a unique trajectory $x_{x_0}^{u,v}\colon [0,t] \to \X$ that satisfies \eqref{E:nonlinear_ss_eqn} and the initial condition $x_{x_0}^{u,v}(0)=x_0$. When clear from the context, we shall drop the subscript and superscript from the trajectory notation.

We assume that the disturbance input $v$ is unknown but takes values in the set $\V$. This set can also be used to capture any (unknown but bounded) uncertainties or unmodeled nonlinearities in the model.

Because we are seeking optimal input signals (functions of time), we
must address the question of what information is available to each
player when their input value is chosen.  From an implementation
perspective the most natural choice is feedback control, in which each
player knows the current state of the system when choosing their input
value at a given time.  Mathematically, however, in a two player game
there is also the question of the order in which the players declare
their input choices, where the second player may gain an advantage
from knowledge of the first player's choice.  In a discrete time
setting this ordering is easily handled by the order in which the
inputs are quantified when defining the discriminating kernel or
reachable set: the input quantified first can be chosen based on
knowledge of the current state (a feedback policy), while the input
quantified second knows both the current state and the value of the
other input signal.  In our context we wish to conservatively estimate
the discriminating kernel, so we will give any advantage to the
disturbance input; consequently, in the discrete time case the
disturbance input would be quantified second.

Unfortunately, the continuous time setting is more complicated.
Because the control knows the state at all times through feedback, the
disturbance's choice of input (or at least any effect it has on the
evolution of the system) is immediately revealed; consequently, the
control is able to detect and respond to disturbance choices
arbitrarily fast even if it is chosen ``first''.  Mathematically, this
situation is dealt with through the construct of non-anticipative
strategies for the disturbance~\cite{Evans_Souganidis_1984}: A map
$\varrho \colon \Ul^\CL_{[0,t]} \to \Vl_{[0,t]}$ is non-anticipative
for $v$ if for every $u(\cdot), u'(\cdot) \in \Ul^\CL_{[0,t]}$, $u(s)
= u'(s)$ implies $\varrho[u](s) = \varrho[u'](s)$ a.e.\ $s \in [0,t]$.
To avoid notational clutter, we denote the class of non-anticipative
disturbance signals $v(\cdot)=\varrho[u](\cdot)$ over $[0,t]$ by
$\Vl^\NA_{[0,t]}$.  Consider a compact state constraint set $\K
\subset \X$ with nonempty interior. We seek to synthesize feedback
control policies that keep system \eqref{E:nonlinear_ss_eqn} contained
in $\K$ over the entire time interval $[0,\tau]$.

\begin{defn}[Safety-Preserving Control]
      With $\K$ deemed ``safe'', for $x_0 \in \K$, a control law $u(\cdot)\in \Ul^\CL_{[0,\tau]}$ is safety-preserving over $[0,\tau]$ if for every disturbance $v(\cdot)\in \Vl^\NA_{[0,\tau]}$ the closed-loop trajectory emanating from $x_0$ remains in $\K$ for all time $t\in [0,\tau]$. 
\end{defn}
The subset of states in $\K$ for which such control policies exist is
defined next.  The fact that the control can respond arbitrarily fast
to the disturbance's choice is taken into account by requiring the
disturbance reveal its non-anticipative strategy mapping before the
control chooses its input.
\begin{defn}[Discriminating Kernel] \label{D:disc}
    The finite-horizon discriminating kernel of $\K$ is the set of all initial states in $\K$ for which there exists a safety-preserving control: 
    \begin{equation}\label{E:disc_setbuilder}
        \Disc_{[0,\tau]}(\K,\U,\V):=\bigl\{ x_0 \in \K \mid  \forall v(\cdot) \in \Vl^\NA_{[0,\tau]}, \, \exists  u(\cdot)\in \Ul^\CL_{[0,\tau]}, \, \forall  t\in [0,\tau], \, x_{x_0}^{u,v}(t) \in \K \bigr\}.
    \end{equation}
\end{defn}
While the control chooses its input with knowledge of the
disturbance's response because the disturbance's strategy is
quantified first in~\eqref{E:disc_setbuilder}, it can be shown that the
disturbance still holds a slight
advantage~\cite{Evans_Souganidis_1984}.  Note that this order of
quantification of the control and disturbance inputs agrees with the
reachable set construct used in~\cite{MBT05}. 
For further discussion of the suitability of this information pattern
to reachability and viability, see~\cite{cardaliaguet1999set,MBT05}
and the citations therein.




With $\V=\{0\}$ the discriminating kernel reduces to the (finite horizon) \emph{viability kernel} under the dynamics $\dot{x}=f(x,u)$:
    \begin{equation}\label{E:viab_setbuilder}
        \Viab_{[0,\tau]}(\K,\U):= \Disc_{[0,\tau]}(\K,\U,\{0\}) = \bigl\{ x_0 \in \K \mid  \exists  u(\cdot)\in \Ul_{[0,\tau]}, \,  \forall  t\in [0,\tau], \, x_{x_0}^{u}(t) \in \K \bigr\}.
    \end{equation}


The discriminating kernel as well as the corresponding safety-preserving control strategies have historically been computed using versatile and powerful Eulerian methods such as ~\cite{cardaliaguet1999set, MBT05, gao2006-reachability}. However, the computational complexity of these algorithms is exponential in the number of states. Their applicability, therefore, is limited to systems of low dimensionality (usually less than four). In this paper we wish to synthesize safety-preserving control laws in a more scalable fashion. First, we will present an alternative characterization of the discriminating kernel by extending our viability kernel results from ~\cite{kaynama_HSCC2012}, which will allow for an efficient computation of this set and ultimately the corresponding control laws.

\subsection{Preliminaries}

We say that the vector field $f$ is bounded on $\K$ if there exists a real number $M > 0$ such that $\norm{f(x,u,v)} \leq M$ $\forall (x,u,v) \in \K \times \U \times \V$ for some norm $\norm{\cdot}$. If $\K$ is compact, every continuous vector field $f$ is bounded on $\K$.

A partition $P = \{t_0, t_1, \ldots, t_n\}$ of $[0,\tau]$ is a set of distinct points $t_0, t_1, \ldots, t_n \in [0,\tau]$ with $t_0 = 0$, $t_n = \tau$ and the ordering $t_0 < t_1 < \cdots < t_n$. We denote the number $n$ of intervals $[t_{k-1},t_k]$ in $P$ by $\abs{P}$, the size of the largest interval by $\norm{P} := \max_{k=1}^{\abs{P}} \{t_{k+1}-t_{k}\}$, and the set of all partitions of $[0,\tau]$ by $\mathscr{P}([0,\tau])$.

For a signal $u \colon [0,\tau] \to \U$ and a partition $P = \{t_0, \ldots, t_n\}$ of $[0,\tau]$, we define the tokenization of $u$ corresponding to $P$ as the set of functions $\left\{u_k \colon [0,t_{k}-t_{k-1}] \to \U\right\}$ such that $u_k(t) = u(t+t_{k-1})$. Conversely, for a set of functions $\left\{u_k \colon [0,t_{k}-t_{k-1}] \to \U\right\}$, we define their concatenation $u \colon [0,\tau] \to \U$ as $u(t) = u_{k}(t-t_{k-1})$ if $t\neq 0$ and $u(0)=u_1(0)$, where $k$ is the unique integer such that $t \in (t_{k-1},t_k]$.



The $\norm{\cdot}$-distance of a point $x \in \X$ from a nonempty set $\A \subset \X$ is defined as
\begin{equation}
    \dist(x,\A) := \inf_{ a \in \A}\norm{x-a}.
\end{equation}
For a fixed set $\A$, the map $x \mapsto \dist(x,\A)$ is continuous.

\subsubsection{Additional Notation}
The inner product in a Hilbert space $\X$ is denoted by $\inner{\cdot,\cdot}$. The \emph{Minkowski sum} of any two nonempty subsets $\Y$ and $\Z$ of $\X$ is $\Y \oplus \Z := \{ y+z \mid y \in \Y,\, z\in \Z \}$; their \emph{Pontryagin difference} (or, the \emph{erosion} of $\Y$ by $\Z$) is $\Y \ominus \Z := \{y \mid y \oplus \Z \subseteq \Y \}$. We denote by $\inter{\Y}$ the interior of $\Y$, by $\partial \Y$ its boundary, and by $\cl \Y:=\inter{\Y}\cup \partial \Y$ its closure. The set $\B(\delta):= \{x \mid \norm{x}\leq \delta\}$ denotes the closed $\norm{\cdot}$ norm-ball of radius $\delta \in \Real^+$ about the origin. The sets $2^\Y$ and $\Y^c$ denote the power set and the complement of $\Y$ in $\X$, respectively.

\section{Abstract Algorithms for Safety Analysis \& Control Synthesis}\label{S:approximation_control_phil}

We first show that we can approximate $\Disc_{[0,\tau]}(\K,\U,\V)$ by considering a nested sequence of sets that are robustly reachable in small sub-intervals of $[0,\tau]$. As compared to the viability kernel results in \cite{kaynama_HSCC2012} the generated under-approximation is robust with respect to an adversarial disturbance input.  We then use this approximation to infer the associated safety-preserving control laws.


\begin{defn}[Robust Maximal Reachable Set]\label{D:robust_max_reach_set}
    The robust (backward) maximal reachable set at time $t$ is the set of initial states for which for every non-anticipative disturbance there exists a feedback control such that the trajectories emanating from those states reach $\K$ exactly at $t$:
    \begin{equation}
            \Reach^{\sharp}_t(\K,\U,\V):= \{ x_0 \in \X \mid \forall v(\cdot) \in \Vl^\NA_{[0,t]}, \, \exists u(\cdot) \in \Ul^\CL_{[0,t]},\,  x_{x_0}^{u,v}(t) \in \K \}.
    \end{equation}
\end{defn}

Again, the justification as to the order of quantification of the two inputs in this \emph{continuous-time} framework follows that of the discriminating kernel discussed immediately after Definition~\ref{D:disc}.

\subsection{Safety Analysis: Under-Approximating the Discriminating Kernel}

Given $P\in \mathscr{P}([0,\tau])$, consider the following recursion 
\begin{subequations}\label{E:recursion_CT_Disc}
    \begin{align}
        K_{\abs{P}} (P) &= \K_\downarrow(P),\label{E:recursion_CT_sub1}\\
        K_{k-1} (P) &= \K_\downarrow(P) \cap \Reach^\sharp_{t_{k} - t_{k-1}} (K_k (P),\U,\V) \quad \text{for} \quad k \in \{1, \ldots, \abs{P}\},
    \end{align}
\end{subequations}
where $\K_\downarrow(P) := \{x \in \K \mid \dist(x, \K^c) \geq M \norm{P}\}$ is a subset of $\K$ deliberately chosen at a distance $M \norm{P}$ from its boundary.

\begin{prop}\label{P:continuous_subset_Disc}
    Suppose that the vector field $f$ is bounded on $\K$ by $M$. Then for any partition $P\in \mathscr{P}([0,\tau])$ the final set $K_0(P)$ defined by the recursion \eqref{E:recursion_CT_Disc} satisfies
    \begin{equation}\label{E:disc[0,tau]}
        K_{0}(P) \subseteq \Disc_{[0,\tau]}(\K,\U,\V).
    \end{equation}
\end{prop}

\begin{proof}
    The proof is provided in the Appendix.
\end{proof}

\begin{cor}[Intermediate Discriminating Kernels]
    For any partition $P$ and every $k\in \{1,\dots,\abs{P}+1\}$ the sets $K_{k-1}(P)$ defined by \eqref{E:recursion_CT_Disc} satisfy
    \begin{equation}\label{E:disc[0,tau-t]}
    \begin{split}
        K_{k-1}(P) \subseteq \Disc_{[0,\tau-t_{k-1}]}(\K,\U,\V).
    \end{split}
    \end{equation}
    %
\end{cor}

\begin{rem}
    For a sufficiently large horizon $\tau$, if $K_{k-1}(P) \equiv K_k(P)$ for some $k\in \{1,\dots,\abs{P}\}$ with $t_k \in (0,\tau]$, then $K_0(P)=K_k(P)$ is a subset of the infinite-horizon discriminating kernel $\Disc_{\Real^+}(\K,\U,\V)$ (also known as the largest robust controlled-invariant set \cite{blanchini2008set}).
\end{rem}

The approximation can be made arbitrarily precise by choosing a sufficiently fine partition of time: 
\begin{prop}\label{P:precision_Disc}
    Suppose that the vector field $f$ is bounded on $\K$ by $M$. Then we have
    \begin{equation}\label{E:precision_Disc}
        \Disc_{[0,\tau]} (\inter{\K},\U,\V) \subseteq  \bigcup_{P \in \mathscr{P}([0,\tau])} K_0(P) \subseteq \Disc_{[0,\tau]} (\K,\U,\V).
    \end{equation}
    %
\end{prop}

\begin{proof}
    The proof is provided in the Appendix.
\end{proof}

By approximating $\Disc_{[0,\tau]} (\K,\U,\V)$ using the sub-interval maximal reachable sets via recursion \eqref{E:recursion_CT_Disc} we enable the use of scalable Lagrangian techniques for the computation of the discriminating kernel for high-dimensional systems. As we shall see, this will also allow us to synthesize the safety-preserving control laws in a more scalable manner than using the conventional techniques.

\subsection{Control Synthesis: A Safety-Preserving Strategy}

Our approach to synthesizing safety-preserving controllers is based on utilizing the sets computed offline during the analysis phase (i.e.\ when computing an under-approximation of the discriminating kernel). The following corollary forms the basis of our safety-preserving feedback policy. It follows directly from Proposition~\ref{P:continuous_subset_Disc}, its proof, and the recursion~\eqref{E:recursion_CT_Disc}.

\begin{cor}\label{Cor:safety_preserving_general}
     Suppose $K_0(P)\neq \emptyset$ for a fixed time partition $P\in \mathscr{P}([0,\tau])$. For any initial condition $x_0\in K_0(P)$ the concatenation $u(\cdot)$ of sub-interval control inputs $u_k(\cdot) \in \Ul^\CL_{[0, t_k-t_{k-1}]}$ corresponding to robust maximal reachable sets $\Reach^{\sharp}_{t_k-t_{k-1}}(K_k(P),\U,\V)$ for $k=1,\dots,\abs{P}$ is a safety-preserving control law that, for every disturbance, keeps the trajectory $x(\cdot)$ of the system \eqref{E:nonlinear_ss_eqn} with $x(0)=x_0$ contained in $\K$ for all time $t\in [0,\tau]$.
\end{cor}

Corollary~\ref{Cor:safety_preserving_general} asserts that any method that facilitates the synthesis of robust maximal reachability controllers (e.g.\ \cite{Kostousova01, Girard2008, GGM06, Kurzhanski_Mitchell_Varaiya_2006, Nenninger_Frehse_Krebs_2002}) can be employed to form a safety-preserving policy. One such Lagrangian method is the ellipsoidal techniques \cite{Kurzhanski2000a} implemented in Ellipsoidal Toolbox (ET) \cite{KV06}. In the remainder of the paper we focus on these techniques to develop scalable analysis and synthesis algorithms for LTI systems.\footnote{In \cite{kaynama_Aut2013} we proposed a scalable Lagrangian method based on support vectors for the approximation of the viability kernel, and showed that in the context of discrete-time systems, it can outperform the ellipsoidal techniques. However, the support vector method has not yet been extended to handle (a) continuous-time systems, (b) uncertainties/disturbances, or (c) control synthesis. In this paper we employed the ellipsoidal techniques since they allows for all of these requirements.}


\section{Analysis Algorithm: Offline Phase}\label{S:analysis_algorithm}


Consider the case in which \eqref{E:nonlinear_ss_eqn} is an LTI system
\begin{equation}\label{E:linear_ss}
    \dot{x} = Ax+Bu+Gv
\end{equation}
with appropriately sized constant matrices $A$, $B$, and $G$. We will further assume that the constraints $\K$, $\U$, and $\V$ are (or can be closely and correctly approximated by) nonempty compact ellipsoids.

In \cite{kaynama_HSCC2012} we introduced a scalable piecewise ellipsoidal algorithm for under-approximating the viability kernel $\Viab_{[0,\tau]}(\K,\U)$ based on the ellipsoidal techniques \cite{Kurzhanski2000a} for maximal reachability. That result can easily be extended so that the generated set approximates the discriminating kernel $\Disc_{[0,\tau]}(\K,\U,\V)$ by simply computing the intermediate maximal reachable sets for adversarial inputs in line with the analysis in Proposition~\ref{P:continuous_subset_Disc} and its proof. Here we summarize this extension.


\subsection{Piecewise Ellipsoidal Approximation of $\Disc_{[0,\tau]}(\K,\U,\V)$}

\begin{defn}
    An \emph{ellipsoid} in $\Real^n$ is defined as
    \begin{equation}\label{E:ellipsoid}
        \E(q,Q):= \left\{z \in \Real^n \mid \inner{(z-q), Q^{-1}(z-q)} \leq 1 \right\}
    \end{equation}
    with center $q\in\Real^n$ and symmetric positive definite shape matrix $Q \in \Real^{n\times n}$. A \emph{piecewise ellipsoidal} set is the union of a finite number of ellipsoids.
\end{defn}

Given $P\in \mathscr{P}([0,\tau])$ and some $k\in \{1,\dots,\abs{P}\}$, let $K_k(P)=\E(\tilde{x},\widetilde{X})$ be an ellipsoid with center $\tilde{x}$ and shape matrix $\widetilde{X}$. As in \cite{Kurzhanski_SCL_2000}, with $\M \subseteq \{l \in \X \mid \inner{l,l} =1 \}$ a (possibly finite) subset of all unit-norm vectors in $\X$ and $\delta:=t_k-t_{k-1}$, we have for all $t\in[0,\delta]$
%
%
\begin{equation}\label{E:reach_ellps_subset}
     \bigcup_{\ell_\delta \in \M} \E(x^c(t), X_\ell(t)) \subseteq \Reach^\sharp_{\delta-t}(K_k(P),\U,\V), 
\end{equation}
where $x^c(t)$ and $X_\ell(t)$ are the center and the shape matrix of the \emph{internal approximating} ellipsoid at time $t$ that is tangent to $\Reach^\sharp_{\delta-t}(K_k(P),\U,\V)$ in the direction $\ell(t)\in \X$. For a fixed terminal direction $\ell_\delta \in \M$ with $\ell(\delta) = \ell_\delta$, the direction $\ell(t)$ is obtained from the adjoint equation $\dot{\ell}(t)=-\tr{A}\ell(t)$. The center $x^c(t)$ (with $x^c(\delta)=\tilde{x}$) and the shape matrix $X_\ell(t)$ (with $X_\ell(\delta) = \widetilde{X}$) are determined from differential equations described in \cite{Kurzhanski2002}. We denote the reachable set at $t=0$ computed for a single direction $\ell_\delta$ by $\Reach_\delta^{\sharp[\ell_\delta]}(K_k(P),\U,\V)$. Clearly, this set is an \emph{ellipsoidal subset} of the actual reachable set.


Let $\ellpfunc\colon 2^\X \to 2^\X$ denote a set-valued function that maps a nonempty set to its maximum volume inscribed ellipsoid. For a fixed terminal direction $\ell_\tau \in \M$, the recursion
%
\begin{equation}\label{Alg:PWE_recursion}
\begin{split}
    K_{k-1}^{[\ell_\tau]}(P) = \ellpfunc(K_{\abs{P}}(P) &\cap \Reach^{\sharp[\ell_\tau]}_{t_k-t_{k-1}}(K_k^{[\ell_\tau]}(P),\U,\V))\\
    &\text{for} \quad k\in\{1,\dots,\abs{P}\}
\end{split}
\end{equation}
with base case $K_{\abs{P}}^{[\ell_\tau]}(P)=K_{\abs{P}}(P) = \K_\downarrow(P)$ generates an ellipsoidal set $K_0^{[\ell_\tau]}(P)$ such that
\begin{equation}\label{Alg:PWE_endset}
    \bigcup_{\ell_\tau\in\M} K_0^{[\ell_\tau]}(P) =: K_0^\cup(P) \subseteq \Disc_{[0,\tau]}(\K,\U,\V).
\end{equation}
The set $K_0^\cup(P)$ is therefore a \emph{piecewise ellipsoidal under-approximation} of the discriminating kernel (cf.\ \cite{kaynama_HSCC2012,kaynama_thesis} for more detail). Here, $\K_\downarrow(P)$ can either be formed by eroding $\K$ by $\B(M\norm{P})$, or (to achieve less conservatism) via the procedure described in \cite[Remark 6.1]{kaynama_thesis}.



If they are needed, the intermediate discriminating kernels \eqref{E:disc[0,tau-t]} are under-approximated via
\begin{equation}\label{E:intermediate_Disc}
    \bigcup_{\ell_\tau\in\M} K_{k-1}^{[\ell_\tau]}(P) =: K_{k-1}^\cup(P) \subseteq \Disc_{[0,\tau-t_{k-1}]}(\K,\U,\V).
\end{equation}

\subsubsection{Approximation Error}

Approximation of the intersection in \eqref{E:recursion_CT_Disc} with its maximum volume inscribed ellipsoid in \eqref{Alg:PWE_recursion} results in a loss of accuracy in the computed under-approximation of the discriminating kernel. The choice of partition $P$, i.e.\ the discretization of $[0,\tau]$, affects this accuracy loss significantly. In particular, we have found in experiments (see ~\cite{kaynama_HSCC2012}) that a denser choice appears to yield a more accurate approximation of the kernel. This is likely due to the fact that with a finer partition (i) the eroded set $\K_\downarrow(P)$ is larger, and (ii) the reachable sets change very little over each time steps. Thus, although a larger number of intersections needs to be performed, the intersection error at every iteration becomes smaller which yields a reduced overall accumulated error.

While the volume of the true intersection set may prove difficult to compute, the volume of its maximum volume inscribed and minimum volume circumscribed \cite{ros2002ellipsoidal} ellipsoids can easily be calculated from the determinant of their shape matrices. This provides a computable estimate of the accuracy loss at every iteration since the volume of the error is bounded above by the difference between the volume of its extremal ellipsoids. An adaptive partitioning algorithm could in principle be designed to manage this error, although we have not yet pursued this direction of research.

\subsection{Controlled-Invariance}

While it can be shown (via a simple induction) that in recursion~\eqref{E:recursion_CT_Disc}
\begin{equation}
    K_{\abs{P}}(P) \cap \Reach^{\sharp}_{t_k-t_{k-1}}(K_k(P),\U,\V)
    = K_k(P) \cap \Reach^{\sharp}_{t_k-t_{k-1}}(K_k(P),\U,\V),
\end{equation}
in the piecewise ellipsoidal algorithm~\eqref{Alg:PWE_recursion} this equality does not hold due to the operator $\ellpfunc(\cdot)$: for any two sets $\Y$ and $\Z$, $\Y \subseteq \Z \not \Rightarrow \ellpfunc(\Y) \subseteq \ellpfunc(\Z)$. Therefore in general,
\begin{equation}
    \ellpfunc(K_{\abs{P}}(P) \cap \Reach^{\sharp[\ell_\tau]}_{t_k-t_{k-1}}(K_k^{[\ell_\tau]}(P),\U,\V))
    \neq \ellpfunc(K_k^{[\ell_\tau]}(P) \cap \Reach^{\sharp[\ell_\tau]}_{t_k-t_{k-1}}(K_k^{[\ell_\tau]}(P),\U,\V))
\end{equation}
Consequently, in contrast to the recursion~\eqref{E:recursion_CT_Disc} for which the necessary and sufficient condition $K_k(P) \subseteq \Reach^{\sharp}_{t_k-t_{k-1}}(K_k(P),\U,\V)$ for some $k$ implies convergence, no similar condition exists for the piecewise ellipsoidal algorithm~\eqref{Alg:PWE_recursion}. 



\begin{defn}[Robust Controlled-Invariant Set]\label{D:robust_cntrl_inv}
    A set $\C\subset \X$ is robust controlled-invariant (or simply, controlled-invariant) under \eqref{E:nonlinear_ss_eqn} if for every non-anticipative disturbance there exists a control that renders $\C$ invariant (meaning $x_0\in \C  \Rightarrow x(t)\in \C$ $\forall t > 0$) under the corresponding closed-loop dynamics.
\end{defn}

Consider the recursions~\eqref{E:recursion_CT_Disc} and \eqref{Alg:PWE_recursion}. Unlike the discrete-time case \cite{blanchini2008set}, even if a fixed-point exists it is not necessarily controlled-invariant. However, the following sufficient condition (stated here only in terms of recursion~\eqref{Alg:PWE_recursion}) holds.
\begin{prop}[Controlled-Invariance Condition]\label{P:control_inv_cond}
For a fixed $P\in \mathscr{P}([0,\tau])$, if there exist $k\in\{1,\dots,\abs{P}\}$ and $\ell_\tau \in \M$ such that
\begin{equation}\label{E:control_inv_cond}
    K_{k}^{[\ell_\tau]}(P) \subseteq \Reach^{\sharp[\ell_\tau]}_{t_k-t_{k-1}}(K_k^{[\ell_\tau]}(P),\U,\V),
\end{equation}
then the set
\begin{equation}
    \bigcup_{t\in [t_{k-1},t_k]} \Reach^{\sharp[\ell_\tau]}_{t_k-t}(K_k^{[\ell_\tau]}(P),\U,\V)
\end{equation}
is robust controlled-invariant. Furthermore, the set
\begin{equation}
    K_k^{[\ell_\tau]}(P) \oplus \B(M\abs{t_k-t_{k-1}})
\end{equation}
contains all corresponding closed-loop trajectories (although it is not itself invariant).
\end{prop}
\begin{proof}
    Suppose such $k$ and $\ell_\tau$ exist. Let $\delta := t_k - t_{k-1} > 0$ and define a shorthand notation $\R_\text{inv}:=\bigcup_{t\in [t_{k-1},t_k]} \Reach^{\sharp[\ell_\tau]}_{t_k-t}(K_k^{[\ell_\tau]}(P),\U,\V)$.
    By the semi-group and finite-time controlled-invariance properties of the maximal reachable tube \cite{kurzhanskii96ellipsoidal} we have that $\forall x_0 \in \R_\text{inv}$ $\exists \delta^* \in [0, \delta]$ $\forall v_1 \in \Vl^\NA_{[0,\delta^*]}$ $\exists u_1 \in \Ul^\CL_{[0,\delta^*]}$ $\forall t\in [0,\delta^*]$ $x(t) \in \R_\text{inv} \wedge x(\delta^*)\in K_k^{[\ell_\tau]}(P)$.
    But via \eqref{E:control_inv_cond}, $K_k^{[\ell_\tau]}(P) \subseteq \Reach^{\sharp[\ell_\tau]}_{t_k-t_{k-1}}(K_k^{[\ell_\tau]}(P),\U,\V) \subseteq \R_\text{inv}$. Thus, $x(\delta^*) \in \Reach^{\sharp[\ell_\tau]}_{t_k-t_{k-1}}(K_k^{[\ell_\tau]}(P),\U,\V)$.
    Reset time and let $x_0 := x(\delta^*)$. We know for $x_0$ that $\forall v_2 \in \Vl^\NA_{[0,\delta]}$ $\exists u_2 \in \Ul^\CL_{[0,\delta]}$ $\forall t \in [0,\delta]$ $x(t) \in \R_\text{inv} \wedge x(\delta)\in K_k^{[\ell_\tau]}(P)$.
    Performing this reset recursively with $v_i$ and $u_i$ for $i=3,\dots$ we can concatenate the inputs to state that $\forall x_0 \in \R_\text{inv}$ $\forall v = \{v_i\}_{i=1}^\infty \in \Vl^\NA_{\Real^+}$ $\exists u = \{u_i\}_{i=1}^\infty \in \Ul^\CL_{\Real^+}$ $\forall t\in [0, \delta^* + \sum_{i=2}^\infty \delta)=\Real^+$ $x(t)\in \R_\text{inv}$.
    The second claim is clear from the fact that over the sub-interval $[t_{k-1},t_k]$ the trajectories can only travel a maximum distance of $M\abs{t_k-t_{k-1}}$ away from the target $K_k^{[\ell_\tau]}(P)$; see \eqref{E:max_dist_travel}.
\end{proof}

%

We will later show that this result can be used for synthesis of \emph{infinite-horizon} safety controllers without the need for convergence (or even termination) of algorithm~\eqref{Alg:PWE_recursion} during the offline phase.

\section{Synthesis Algorithm: Online Phase}\label{S:synthesis_algorithm}

Let the control set $\U$ be a compact ellipsoid with center $\mu \in \Real^{m_u}$ and shape matrix $U\in \Real^{m_u \times m_u}$:
\begin{equation}
    \U=\E(\mu, U).
\end{equation}
Based on the piecewise ellipsoidal algorithm described above, we can form a safety-preserving feedback policy taking values in $\U$ that keeps the trajectory of the system in $\K$ over the entire time horizon (despite the actions of the disturbance). We do this by concatenating the sub-interval robust maximal reachability control laws according to Corollary~\ref{Cor:safety_preserving_general}.

Suppose that while computing the maximal reachable sets $\Reach^{\sharp[\ell_\tau]}_{t_k-t_{k-1}}(K_k^{[\ell_\tau]}(P),\U,\V)$ in \eqref{Alg:PWE_recursion} we can also compute
\begin{equation}
    \Reach^{\sharp[\ell_\tau]}_{t_k-t}(K_k^{[\ell_\tau]}(P),\U,\V) \quad \forall t\in (t_{k-1},t_k).
\end{equation}
That is, we compute the entire reachable \emph{tube} over the $k$th sub-interval. For fixed $\ell_\tau\in\M$ and $k\in\{1,\dots,\abs{P}\}$ let $x^{c[\ell_\tau]}_k(t-t_{k-1})$ and $X^{[\ell_\tau]}_{\ell,k}(t-t_{k-1})$ denote the center and the shape matrix of the ellipsoid $\Reach^{\sharp[\ell_\tau]}_{t_k-t}(K_k^{[\ell_\tau]}(P),\U,\V)$ at time $t\in[t_{k-1},t_k]$ with $K_k^{[\ell_\tau]}(P)=\E(x^{c[\ell_\tau]}_k(t_k-t_{k-1}),X^{[\ell_\tau]}_{\ell,k}(t_k-t_{k-1}))$ and terminal direction $\ell_\tau$. Define a shorthand notation
%
%
\begin{equation}\label{E:shorthand_notation_R}
    \begin{split}
        \R(t,k) &:= \bigcup_{\ell_\tau\in\M} \R^{[\ell_\tau]}(t,k)\\
                &:=\bigcup_{\ell_\tau\in\M} \Reach^{\sharp[\ell_\tau]}_{t_k-t}(K_k^{[\ell_\tau]}(P),\U,\V)\\
                &\,\, =\bigcup_{\ell_\tau\in\M} \E(x^{c[\ell_\tau]}_k(t-t_{k-1}),X^{[\ell_\tau]}_{\ell,k}(t-t_{k-1})),
    \end{split}
\end{equation}
and suppose $\R(0,1)\neq \emptyset$.

We begin by defining a \emph{pseudo-time} function $\sigma \colon \Real^+ \to [0,\tau]$ with varying rate of change $\dot{\sigma}(t)\in [0,1]$ such that with the initial time $s\in[0,t]$,
\begin{equation}\label{E:pseudotime_defn}
    \sigma(t)= s + \int_s^t \dot{\sigma}(\lambda)d\lambda.
\end{equation}
%
For a fixed initial time $s$, this function is continuous and non-decreasing. However, while the global time $t$ increases strictly monotonically, $\sigma$ increases with varying rates as determined by $\dot{\sigma}$. Clearly we can see that $\sigma(t)\leq t$. For a fixed $k\in \{1,\dots,\abs{P}\}$ let
%
\begin{equation}
    \sigma(t)\in [\underline{\sigma}_k,\overline{\sigma}_k]=:\mathbb{S}_k 
\end{equation}
with
\begin{equation}
    \underline{\sigma}_k:=t_{k-1} \quad \text{and} \quad \overline{\sigma}_k:=t_k.
\end{equation}
Denote by $\mathbb{T}_k$ a totally ordered subset of $\Real^+$ such that (via an abuse of notation)
\begin{equation}
    \sigma(\mathbb{T}_k)=\mathbb{S}_k
\end{equation}
in the sense that with $\underline{\theta}_k:=\min(\mathbb{T}_k)$ and $\overline{\theta}_k:=\max(\mathbb{T}_k)$,
\begin{equation}
    \sigma(\underline{\theta}_k)=\underline{\sigma}_k \quad \text{and} \quad \sigma(\overline{\theta}_k)=\overline{\sigma}_k.
\end{equation}
%
Note that $\bigcup_{k=1}^{\abs{P}} \mathbb{S}_k = [0,\tau]$, while $\bigcup_{k=1}^{\abs{P}} \mathbb{T}_k \supseteq [0,\tau]$ since $\abs{\mathbb{T}_k} \geq \abs{\mathbb{S}_k}$ $\forall k$.


This pseudo-time variable will be used as part of our formulation of a hybrid automaton that implements the desired safety-preserving controller (Fig.~\ref{F:block_diagram}). The mode of the automaton dictates the rate of change of $\sigma$; therefore, the controller has the possibility of potentially preserving safety over a larger horizon than $\tau$ even though the computations during the offline phase have only been performed over $[0,\tau]$. This hybrid automaton is described next.

\begin{figure}[t]
    \centering
    \scalebox{.5}{\input{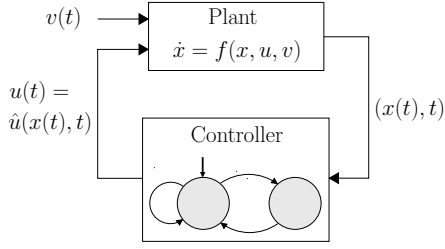}}
    \caption{Block diagram of the closed-loop system employing the safety-preserving controller of Section~\ref{S:hybrid_controller}.}
    \label{F:block_diagram}
\end{figure}

\subsection{The Safety-Preserving Hybrid Controller}\label{S:hybrid_controller}

Following the formalism introduced in ~\cite{Tomlin2000a}, consider the hybrid automaton
\begin{equation}\label{E:automaton_H}
    H=(\Q,\InpEx,\InpIn,\Init,\Dom,E,G,R,\U_\text{fb})
\end{equation}
depicted in Fig.~\ref{F:automaton}
\begin{figure*}[t]
    \centering
    \scalebox{.63}{\input{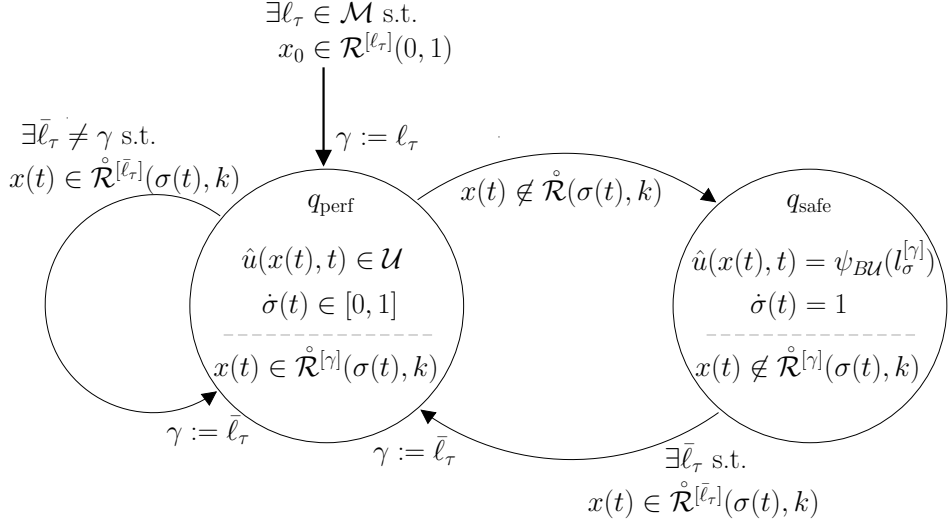}}
    \caption{The graph of the hybrid automaton $H$ representing the safety-preserving controller.}
    \label{F:automaton}
\end{figure*}
where $\Q=\{q_\text{perf}, q_{\text{safe}}\}$ is the set of discrete states with $q_\text{perf}$ representing the case in which the controller is free to choose any value in $\U$ (``performance'' mode) and with $q_\text{safe}$ representing the case in which the controller is required to return an optimal safety-preserving law (``safety'' mode).

The inputs to the controller are drawn from the sets $\InpEx\subseteq \X \times \Real^+$ (\emph{external} input) and $\InpIn \subseteq \M$ (\emph{internal} input). The external input is the pair $(x(t),t)\in \InpEx$ and the internal input is the direction vector $\gamma \in \InpIn$. The initial state of the automaton $\Init \subseteq \Q \times \InpEx \times \InpIn$ is assumed to be $\Init = (q_\text{perf}, x_0, 0, \{\gamma \mid \exists \ell_\tau \in \mathcal{M},\, x_0 \in \R^{[\ell_\tau]}(0,1),\, \gamma=\ell_\tau \})$. The domains $\Dom(\cdot,\gamma,\sigma(t)) \colon \Q \to 2^\X$ of the automaton for every $\gamma\in \InpIn$ and $\sigma(t)\in [0,\tau]$ are $\Dom(q_\text{perf},\gamma, \sigma(t)) = \inter{\R}^{[\gamma]}(\sigma(t),k)$ and $\Dom(q_\text{safe},\gamma, \sigma(t)) = (\Dom(q_\text{perf},\gamma, \sigma(t)))^c$. The domains for every $(x(t),t)\in \InpEx$ specify $(\gamma,\sigma(t))$-varying invariants that must be satisfied in each mode.

The rate of change of $\sigma$ varies in $q_\text{perf}$ such that $\dot\sigma(t)\in[0,1]$ with $\dot\sigma(t)=0$ indicating a pseudo-time freeze. This allows the controller to maintain its invariant/domain in $q_\text{perf}$ for as long as possible by freezing (or slowing down) the progression of the pseudo-time. On the other hand, in $q_\text{safe}$ the pseudo-time $\sigma$ is forced to change with the same rate as the global time $t$ since the state is already on the boundary (or the exterior) of all possible domains of $q_\text{perf}$ and thus $\sigma$ must progress with $t$ in order for $H$ to preserve safety. As we shall see, this varying rate of change can ensure safety over a longer time horizon than $\tau$.

The edges $E\subseteq \Q\times \Q$ are $E=\{(q_\text{perf},q_\text{perf}), \allowbreak (q_\text{perf},q_\text{safe}),\allowbreak (q_\text{safe}, q_\text{perf})\}$. The guards $G(\cdot,\gamma,\sigma(t))\colon E \to 2^\X$ for every $\gamma\in \InpIn$ and $\sigma(t)\in [0,\tau]$ are conditions on $(x(t),t)\in\InpEx$ defined as:
$G((q_\text{perf},q_\text{safe}),\gamma,\sigma(t))=(\inter{\R}(\sigma(t),k))^c$, $G((q_\text{safe}, q_\text{perf}),\gamma,\sigma(t))= \{\inter{\R}^{[\bar\ell_\tau]}(\sigma(t),k) \allowbreak\: \allowbreak\text{for some}\; \allowbreak \bar\ell_\tau\in\InpIn\}$, and $G((q_\text{perf}, q_\text{perf}),\gamma,\sigma(t))=\{\inter{\R}^{[\bar\ell_\tau]}(\sigma(t),k) \allowbreak\: \allowbreak\text{for some}\; \allowbreak \bar\ell_\tau\in\InpIn,\, \bar\ell_\tau\neq \gamma\}$.
The domains and the guards are chosen in terms of the \emph{interior} of the sets $\R^{[\gamma]}(\sigma(t),k)$ to ensure that the automaton $H$ is non-blocking and that transitions over $E$ can take place when necessary. A transition corresponding to an edge is enabled for every $t$ if $x(t)$ satisfies its guard. For example, the automaton can make a transition from $q_\text{safe}$ to $q_\text{perf}$ over the edge $(q_\text{safe},q_\text{perf})\in E$ for a fixed $(\gamma, x(t),t) \in \InpIn \times \InpEx$ if $\exists \bar\ell_\tau$ (including $\bar\ell_\tau=\gamma$) s.t.\ $x(t)\in \inter{\R}^{[\bar\ell_\tau]}(\sigma(t),k)$.

The reset map $R\colon E \times \InpIn \to \InpIn$ resets the internal input via $R(q_\text{perf}, q_\text{perf},\gamma)=\bar \ell_\tau$, $R(q_\text{perf}, q_\text{safe},\gamma)=\gamma$, and $R((q_\text{safe}, q_\text{perf}),\gamma)=\bar \ell_\tau$. Notice that if $\ell_\tau = \gamma$ satisfies the guard on $(q_\text{safe}, q_\text{perf})$, the reset map on the edge $(q_\text{safe}, q_\text{perf})$ is simply $R(q_\text{safe}, q_\text{perf},\gamma)=\gamma$.

Finally, the output of $H$ is a set-valued map $\U_\text{fb}\colon \Q \times \InpIn \times \InpEx \to 2^{\U}$ given by
\begin{align}
\begin{cases}
    \U_\text{fb}(q_\text{perf},\gamma,x(t),t)=\U;\\
    \U_\text{fb}(q_\text{safe},\gamma,x(t),t)=\psi_{B\U}(l^{[\gamma]}_\sigma),
\end{cases}
\end{align}
where
\begin{align}
    \psi_{B\U}\colon \X &\to \U, \notag\\
    l^{[\gamma]}_\sigma &\mapsto \mu - U\tr{B}l^{[\gamma]}_\sigma \Inner{l^{[\gamma]}_\sigma,B U\tr{B}l^{[\gamma]}_\sigma}^{-1/2}\label{E:psi}
\end{align}
is chosen so that $B\psi_{B\U}(l^{[\gamma]}_\sigma)$ is the support vector of the set $B\U=\E(B \mu, B U \tr{B})\subseteq \X$ in the direction $-l^{[\gamma]}_\sigma\in \X$ with
\begin{equation}\label{E:control_l(x(t))}
    l^{[\gamma]}_{\sigma}=(X^{[\gamma]}_{\ell,k}(\sigma(t)-\underline{\sigma}_k))^{-1}(x(t)-x^{c[\gamma]}_{k}(\sigma(t)-\underline{\sigma}_k)).
\end{equation}
%
%
This strategy is based on the optimal control design presented in \cite{kurzhanskii96ellipsoidal} and \cite{Kurzhanskiy_Varaiya_2008}.

We allow non-determinism in the executions of the hybrid automaton to formulate a non-restrictive policy (in the sense of \cite{Tomlin2000a}); $q=q_\text{safe}$ only when safety is at stake, and $q=q_\text{perf}$ otherwise. As we shall see, the primary objective of the controller---to preserve safety---is achievable regardless of this behavior. 

\begin{thm}[Safety-Preserving Controller]\label{Thm:safety_preserving_control}
    For a given partition $P\in \mathscr{P}([0,\tau])$ for any $x_0 \in K_0^\cup(P)$ where $K_0^\cup(P)$ is the piecewise ellipsoidal set obtained through \eqref{Alg:PWE_recursion}--\eqref{Alg:PWE_endset}, there exists a time $T \geq \tau$, $\sigma(T)=\tau$, such that for any non-anticipative disturbance $v$ with $v(t) \in \V$ the (non-deterministic) feedback policy
    \begin{equation}\label{E:closedloop_control}
        u(t) = \hat{u}(x(t),t) \in \U_{\textup{fb}}
    \end{equation}
    generated by the automaton $H$ keeps the trajectory $x(\cdot)$ of the system \eqref{E:linear_ss} with initial condition $x(0)=x_0$ contained in $\K$ for all time $t\in[0,T]$.
\end{thm}

\begin{proof}
    We prove that safety is preserved in each mode for any given $(\gamma,x(t),t,\sigma(t))\in \M\times \X\times \mathbb{T}_k \times \mathbb{S}_k$. We will draw upon the maximal reachability results in ~\cite{kurzhanskii96ellipsoidal}.


    First, if for every $k$, $x(\underline{\theta}_k)\in K_{k-1}^\cup(P)$ (which, as we shall see, is indeed the case) then $x(\underline{\theta}_k) \in \R(\underline{\sigma}_k,k)$ since $K_{k-1}^\cup(P) \subseteq \R(\underline{\sigma}_k,k)$. Fix $\gamma=\ell_\tau\in\M$ and $k$ and let $x(\underline{\theta}_k) \in \R^{[\gamma]}(\underline{\sigma}_k,k)$. Define a continuously differentiable value function $V_{k}^{[\gamma]}\colon \X \times \mathbb{S}_k\to \Real$ such that
    \begin{equation}
        V_{k}^{[\gamma]}(x(t),\sigma(t))=\dist^2(x(t),\R^{[\gamma]}(\sigma(t),k)).
    \end{equation}
    Notice that $V_{k}^{[\gamma]}(x(t),\sigma(t)) = 0$ for $x(t) \in \inter{\R}^{[\gamma]}(\sigma(t),k)$ and $V_{k}^{[\gamma]}(x(t),\sigma(t)) \geq 0$ for $x(t) \not\in \inter{\R}^{[\gamma]}(\sigma(t),k)$. From this definition we use the convention 
    \begin{equation}\label{E:level_set_V}
            \R^{[\gamma]}(\underline{\sigma}_k,k) = \bigl\{x\in\X \mid V_{k}^{[\gamma]}(x,\underline{\sigma}_k)\leq 0\bigr\}.
    \end{equation}

    Consider the case in which $x(t)\not\in \inter{\R}^{[\gamma]}(\sigma(t),k)$ and the active mode of the automaton is $q_\text{safe}$. Clearly, we have that $\dist(x(t),\R^{[\gamma]}(\sigma(t),k))\geq 0$. Calculating the Lie derivative of $V_{k}^{[\gamma]}(x,\sigma(t))$ along the system dynamics \eqref{E:linear_ss} yields
    \begin{equation}\label{E:dVdt_synthesis_proof}
        \frac{d}{dt}V_{k}^{[\gamma]}(x(t),\sigma(t)) = 2 \dist(x(t),\R^{[\gamma]}(\sigma(t),k))
        \times \frac{d}{dt}\dist(x(t),\R^{[\gamma]}(\sigma(t),k)).
    \end{equation}
    Due to strict convexity of the sets, the distance function can be described \cite{kurzhanskii96ellipsoidal} by
    \begin{equation}\label{E:d_max_rho}
        \dist(x(t),\R^{[\gamma]}(\sigma(t),k)) = \max_{\norm{l}\leq 1} \left\{ \Inner{l,x(t)}-\rho_{\R^{[\gamma]}(\sigma(t),k)}(l)\right\}  =  \Inner{l^{[\gamma]}_{\sigma},x(t)}-\rho_{\R^{[\gamma]}(\sigma(t),k)}(l^{[\gamma]}_{\sigma}),
    \end{equation}
    where $\rho_{\A}(l):= \sup_{x\in \A} \inner{l,x}$ is the support function of the convex set $\A\subset \X$ in the direction $l\in \X$, and $l^{[\gamma]}_\sigma$ is the maximizer (unique for $\dist(x(t),\R^{[\gamma]}(\sigma(t),k)) \geq 0$) that is the direction vector determined by \eqref{E:control_l(x(t))}. Thus, following the analysis in \cite{kurzhanskii96ellipsoidal} we have that
    \begin{align}
        \frac{d}{dt}\dist(x(t),\R^{[\gamma]}(\sigma(t),k))
            &= \frac{\partial}{\partial t}\left\{ \Inner{l^{[\gamma]}_{\sigma},x(t)}-\rho_{\R^{[\gamma]}(\sigma(t),k)}(l^{[\gamma]}_{\sigma})\right\} \notag\\
                &\quad + \frac{\partial}{\partial x}\left\{ \Inner{l^{[\gamma]}_{\sigma},x(t)}-\rho_{\R^{[\gamma]}(\sigma(t),k)}(l^{[\gamma]}_{\sigma})\right\} \cdot \dot{x}(t)\\
            &= \Inner{l^{[\gamma]}_{\sigma},\dot{x}(t)} - \frac{\partial}{\partial \sigma} \rho_{\R^{[\gamma]}(\sigma(t),k)}(l^{[\gamma]}_{\sigma}) \cdot \dot{\sigma}(t) \notag\\
                &\quad + \left( l^{[\gamma]}_{\sigma} - \frac{\partial}{\partial x} \rho_{\R^{[\gamma]}(\sigma(t),k)}(l^{[\gamma]}_{\sigma})\right) \cdot \dot{x}(t)\\
            & = \Inner{l^{[\gamma]}_{\sigma},\dot{x}(t)} - \frac{\partial}{\partial \sigma} \rho_{\R^{[\gamma]}(\sigma(t),k)}(l^{[\gamma]}_{\sigma})+0 \qquad \text{($\dot{\sigma}(t)=1$ in $q_\text{safe}$)}\\
            &\leq \Inner{l^{[\gamma]}_{\sigma},Bu(t)+Gv(t)} + \rho_{B\U}(-l^{[\gamma]}_{\sigma})- \rho_{G\V}(l^{[\gamma]}_{\sigma}). \label{E:ddist/dt_for_ellipsoids_extension}
    \end{align}
    We refer the reader to \cite[Section 1.8]{kurzhanskii96ellipsoidal} for additional details, particularly on the reason why $\frac{\partial}{\partial s} \rho_{\R^{[\gamma]}(s,k)}(l) \geq \rho_{G\V}(l) - \rho_{B\U}(-l)$.

    If the disturbance plays optimally (i.e., in a worst-case fashion) then $\Inner{l^{[\gamma]}_{\sigma},Gv(t)} = \rho_{G\V}(l^{[\gamma]}_{\sigma})$ and the right hand side in \eqref{E:ddist/dt_for_ellipsoids_extension} simplifies to
    \begin{equation}\label{E:optimal_v_supportEquations}
            \Inner{l^{[\gamma]}_\sigma,Bu(t)} + \rho_{B\U}(-l^{[\gamma]}_\sigma)
            = \Inner{l^{[\gamma]}_\sigma,Bu(t)} + \sup_{x\in B\U}\Inner{-l^{[\gamma]}_\sigma,x}
            = \Inner{l^{[\gamma]}_\sigma,Bu(t)} - \inf_{x\in B\U}\Inner{l^{[\gamma]}_\sigma,x}.
    \end{equation}
    The set $B\U$ is a compact ellipsoid. Invoking \eqref{E:psi} one can verify from \eqref{E:optimal_v_supportEquations} that under optimal disturbance,
    \begin{equation}\label{E:dist_optimal_disturbance}
            \dfrac{d}{dt}\dist(x(t),\R^{[\gamma]}(\sigma(t),k)) = 0 \quad \text{for} \;\; u(t)=\psi_{B\U}(l^{[\gamma]}_\sigma).
    \end{equation}
    %
    Recall that $l^{[\gamma]}_\sigma$ is the maximizer in \eqref{E:d_max_rho} that is determined for every $t$ based on $x(t)$. The disturbance is non-anticipative since it has all of the information about the state (and thus the past values of the control) plus the current value of the control (since $u(t)$ is chosen based on a fixed $l^{[\gamma]}_\sigma$). On the other hand, the control chooses its value $u(t)=\psi_{B\U}(l^{[\gamma]}_\sigma)$ after the disturbance plays and therefore it is able to keep the derivative of the distance function zero (the derivative would be positive for any non-optimal values of $u$).
    In the general case when $v(t)$ plays non-optimally against the (optimal) control, the inequality \eqref{E:ddist/dt_for_ellipsoids_extension} implies that
    \begin{equation}\label{E:dist_nonoptimal_disturbance}
        \frac{d}{dt}\dist(x(t),\R^{[\gamma]}(\sigma(t),k)) \leq 0 \quad \text{for}\;\; u(t)=\psi_{B\U}(l^{[\gamma]}_{\sigma}).
    \end{equation}

    On the other hand, when $x(t) \in \inter{\R}^{[\gamma]}(\sigma(t),k)$ and the active mode of the controller is $q_\text{perf}$ we have that $\dist(x(t),\R^{[\gamma]}(\sigma(t),k)) = 0$. Thus the derivative of the distance function need not be examined as the Lie derivative of the value function \eqref{E:dVdt_synthesis_proof} is automatically zero regardless of the value of $u(t)$. In addition, in this mode we have that $\dot\sigma(t)\in[0,1]$ and therefore the automaton can generate an execution such that the pseudo-time does not advance.

    Combining the above we see from \eqref{E:dVdt_synthesis_proof} that with $u(t)=\hat{u}(x(t),t)\in \U_{\text{fb}}(q,\gamma,x(t),t)$,
    \begin{equation}
        \frac{d}{dt}V_{k}^{[\gamma]}(x,\sigma(t))\leq 0 \qquad \forall (q,\gamma,x,t,\sigma(t),v)\in \Q\times \M\times \X \times \mathbb{T}_k \times \mathbb{S}_k \times \V.
    \end{equation}
    %
    %
    %
    Integrating from $\underline{\theta}_k$ to $t$ in turn yields (via \eqref{E:level_set_V})
    \begin{equation}
        V_k^{[\gamma]}(x(t),\sigma(t))\leq V_k^{[\gamma]}(x(\underline{\theta}_k),\underline{\sigma}_k)\leq 0. 
    \end{equation}
    Therefore, for every solution $x(\cdot)$ of the differential inclusion $\dot{x}(t)\in Ax(t) \oplus B\U_{\text{fb}}(q,\gamma,x(t),t) \oplus G\V$ we will have
    %
    \begin{equation}
        x(t) \in \R^{[\gamma]}(\sigma(t),k) = \Reach^{\sharp[\gamma]}_{\overline{\sigma}_k-\sigma(t)}(K_k^{[\gamma]}(P),\U,\V)            \qquad \forall (t,\sigma(t)) \in \mathbb{T}_k \times \mathbb{S}_k,
    \end{equation}
    and in particular at $t=\overline{\theta}_k$, $x(\overline{\theta}_k) \in K_k^{[\gamma]}(P)$. According to \eqref{E:intermediate_Disc},
    %
    \begin{equation}\label{E:K_k_in_K}
        K_k^{[\gamma]}(P) \subseteq \bigcup_{\hat{\gamma}\in\M} K_k^{[\hat{\gamma}]}(P) =: K_k^\cup(P)
        \subseteq \Disc_{[0,\tau-\overline{\sigma}_k]}(\K,\U,\V)\subseteq \K.
    \end{equation}
    %
    Therefore, $x(\overline{\theta}_k)\in\K$ for every $k\in\{1,\dots,\abs{P}\}$ granted that the feedback control law \eqref{E:closedloop_control} is applied. It remains to show that $x(t)\in\K$ for all $t\in (\underline{\theta}_k,\overline{\theta}_k)$. Indeed, by construction of the recursion \eqref{Alg:PWE_recursion} via Proposition~\ref{P:continuous_subset_Disc} we have that for all $\sigma(t)\in \mathbb{S}_k$ and $k\in \{1,\dots,\abs{P}\}$
    \begin{equation}
        \Reach^{\sharp[\gamma]}_{\overline{\sigma}_k-\sigma(t)}(K_k^{[\gamma]}(P),\U,\V)  \subseteq \K. 
    \end{equation}
    Therefore we conclude that with $u(t)\in \U_\text{fb}$ and $x_0\in \R(0,1)$ we have $x(t)\in \K$ $\forall t\in [\underline{\theta}_1,\overline{\theta}_{\abs{P}}]$ $\forall v(\cdot) \in \Vl^\NA_{[\underline{\theta}_1,\overline{\theta}_{\abs{P}}]}$. Note that $\underline{\theta}_1 = \underline{\sigma}_1 = 0$ and $\overline{\theta}_{\abs{P}}\geq \tau$. Let $T:= \overline{\theta}_{\abs{P}}$ and recall that $\sigma(\overline{\theta}_{\abs{P}})= \overline{\sigma}_{\abs{P}}=\tau$. Thus $x(t)\in \K$ $\forall t\in [0,T]$ $\forall v(\cdot) \in \Vl^\NA_{[0,T]}$.
\end{proof}

\subsection{Remarks \& Other Considerations}\label{S:control_remarksSection}


\subsubsection{Respecting the Intermediate Kernels}
The strategy \eqref{E:closedloop_control} will ensure that the closed-loop system evolves in the interior or else on the boundary of $\R(\sigma(t),k)$ for all $(t,\sigma(t))\in \mathbb{T}_k \times \mathbb{S}_k$ and $k\in\{1,\dots,\abs{P}\}$. Thus, as we have seen above, for a given $P\in\mathscr{P}([0,\tau])$ even if the disturbance always plays a worst-case game the trajectories satisfy $x(\overline{\theta}_k)\in K^\cup_{k}(P) \subseteq \Disc_{[0,\tau-\sigma(\overline{\theta}_k)]}(\K,\U,\V)$ for all $k$. 

\subsubsection{Shared Boundary Points}
In practice for common points that are on the boundaries of two (or more) $\ell_\tau$-parameterized ellipsoids $\Reach^{\sharp[\ell_\tau]}_{\overline{\sigma}_k-\sigma(t)}(K_k^{[\ell_\tau]}(P),\U,\V)$ at any pseudo-time instant $\sigma(t)$, any one of these ellipsoids can be used for the computation of $l_\sigma^{[\ell_\tau]}$ in \eqref{E:control_l(x(t))} and its corresponding optimal control law $\psi_{B\U}(l_\sigma^{[\ell_\tau]})$ in mode $q_\text{safe}$. Any such control will ensure that the trajectory remains within the corresponding robust maximal reachable tube while being steered towards $K_k^{[\ell_\tau]}(P)$ (which, as we have seen, will ultimately ensure constraint satisfaction and safety).

To allow a more permissive control policy, the automaton $H$ could be modified by adding a self-loop in mode $q_\text{safe}$ (i.e., an edge $(q_\text{safe},q_\text{safe})$) so that among all possible safe control laws corresponding to those ellipsoids that share the common boundary point, one that better satisfies a given performance criterion could be selected.

\subsubsection{Conservatism of $K_0^\cup(P)$ \& Synthesis for $x_0 \not\in K_0^\cup(P)$}\label{S:control_notinK0star}
When the initial condition $x_0$ lies inside the true discriminating kernel but outside its piecewise ellipsoidal under-approximation, i.e.\ $x_0 \in \Disc_{[0,\tau]}(\K,\U,\V)\backslash K_0^\cup(P)$ for a given $P\in \mathscr{P}([0,\tau])$, the control policy \eqref{E:closedloop_control} can no longer guarantee safety. 
However, with the following modification the feedback law can often keep the trajectory within $\K$ in practice:  For every $t\in [0,T]$ such that $x(t) \not \in \R(\sigma(t),k)$ with $x(0) = x_0$ and $\sigma(t) \in \mathbb{S}_k$, the direction vector $l_\sigma^{[\ell_\tau]}$ in \eqref{E:control_l(x(t))} is modified to be the direction that corresponds to any $\ell_\tau \in \M$ for which $\dist(x(t),\Reach^{\sharp[\ell_\tau]}_{\overline{\sigma}_k-\sigma(t)}(K_k^{[\ell_\tau]}(P),\U,\V))$ is minimized; See ~\cite[Section 6.3]{kaynama_thesis} for an example.


Finally, note that for $x_0\not \in \Disc_{[0,\tau]}(\K,\U,\V)$ the modified control law described above \emph{may} still be able to keep $x$ in $\K$ over at least $[0,\tau]$ if the disturbance does not play optimally against the control. On the other hand, in the deterministic case ($\V=\{0\}$), for $x_0\not \in \Viab_{[0,\tau]}(\K,\U)$ there does not exist a control law that can keep the trajectory within $\K$ over $[0,\tau]$. 

\subsubsection{Pseudo-Time and Conservatism}

By introducing the notion of pseudo-time in $H$ the objective of the controller becomes two-fold: 1) to preserve safety by applying a particular optimal control law when (and only when) safety is at stake, and 2) to prolong the duration for which safety can be preserved. This second objective may cause the closed-loop trajectories to be driven towards the boundary of the domain in $q_\text{perf}$. Therefore the overall behavior of the system becomes conservative, particularly if pseudo-time is always frozen in $q_\text{perf}$; the automaton may spend a significant amount of (global) time close to the boundaries of the domains of both modes. Even if chattering can be avoided, this phenomenon makes the safety mode active more often than would have been necessary if the controller were only to pursue the first objective. In such a case, the second objective may be relaxed by choosing $\dot{\sigma}(t)=1$ in $q_\text{perf}$ in favor of performance. Safety is still preserved over at least $[0,\tau]$ with this choice.

\subsubsection{Smooth Transition Between Safety and Performance}\label{S:smooth}

When the mode $q_\text{perf}$ of the controller is active, the strategy $\U_\text{fb}(q_\text{perf},\gamma,x(t),t)=\U$ for some $(\gamma,x(t),t) \in \InpIn \times \InpEx$ and $\sigma(t)\in [0,\tau]$ is returned and any control input in $\U$ can be applied to the system; denote this input as $u_\text{perf}(t)$. On the other hand, when the mode $q_\text{safe}$ is active, the controller returns $\U_\text{fb}(q_\text{safe},\gamma,x(t),t)=\psi_{B\U}(l_\sigma^{[\gamma]})$ and this specific safety-preserving control law must be applied to the system to ensure safety; we denote this input as $u_\text{safe}(t)$.

Choosing $u_\text{perf}$ arbitrarily without considering the main goal of the closed-loop system (preserving safety) may result in excessive switching between the two modes of the controller (whose priority is to preserve safety over performance). Thus the resulting control policy could have a high switching frequency and the controller could end up spending a significant amount of time at the extremum points of the input constraint set. Such a policy may be hard on actuators due to chattering\footnote{We define chattering of a signal or automaton as frequent (but finite) jumps in the signal value or automaton mode over a given time interval.} or Zeno behavior. These pathological cases are usually avoided by imposing certain conditions such as dwell-time \cite{Morse1996} (requiring that the automaton remains in each mode for a non-zero amount of time) or other techniques as discussed in, for example \cite{Johansson1999,Asarin_Bournez_Dang_Maler_Pnueli_2000,filippov1988differential}, but none of these approaches work with $H$ to ensure safety.



Instead, we propose the following heuristic to reduce switching frequency while maintaining safety: An ideal control policy in mode $q_\text{perf}$ should be a combination of both $u_\text{perf}$ and $u_\text{safe}$, even though technically the safety component $u_\text{safe}$ is not needed when the controller is in mode $q_\text{perf}$. We use a simple convex combination of these two inputs: For every $(\gamma,x(t),t)\in \InpIn \times \InpEx$, $\sigma(t)\in[0,\tau]$ and a given domain $\inter{\R}^{[\gamma]}(\sigma(t),k)=\inter{\E}(x^{c[\gamma]}_k(\sigma(t)-\underline{\sigma}_k),X^{[\gamma]}_{\ell,k}(\sigma(t)-\underline{\sigma}_k))$ in $q_\text{perf}$ we shall choose an input such that
\begin{equation}\label{E:smooth1}
        \begin{split}
        u(t) &= \left(1-\beta_\alpha(\phi^{[\gamma]}(x(t),\sigma(t)))\right) u_\text{perf}(t) + \beta_\alpha(\phi^{[\gamma]}(x(t),\sigma(t))) u_\text{safe}(t), 
        \end{split}
\end{equation}
with
\begin{equation}\label{E:Beta_intermsof_phi_alpha}
    \beta_\alpha(\xi) :=
    \begin{cases}
        1 & \text{if \; $\xi \geq 1$};\\
        \frac{1}{1-\alpha}(\xi-\alpha) & \text{if \; $\alpha \leq \xi < 1$};\\
        0 & \text{if \; $\xi < \alpha$},
    \end{cases}
\end{equation}
where $\alpha \in [0,1)$ is a design parameter (Fig.~\ref{F:beta}), and
\begin{align}
    \phi^{[\gamma]} \colon \X \times [0,\tau] &\to \Real^+,\\
    (x(t),\sigma(t)) &\mapsto
    \begin{multlined}[t][6cm]
          \Inner{(x(t)-x^{c[\gamma]}_k(\sigma(t)-\underline{\sigma}_k)),\\ (X^{[\gamma]}_{\ell,k}(\sigma(t)-\underline{\sigma}_k))^{-1} (x(t)-x^{c[\gamma]}_k(\sigma(t)-\underline{\sigma}_k))}.
    \end{multlined}
\end{align}
In $q_\text{perf}$ for fixed $\gamma$ and $\sigma(t)$ we have that $\operatorname{dom}(\phi^{[\gamma]}(\cdot,\sigma(t)))=\inter{\R}^{[\gamma]}(\sigma(t),k)$. Therefore, $\operatorname{range}(\phi^{[\gamma]}(\cdot,\sigma(t)))=[0,1)$. This is true simply because the set $\R^{[\gamma]}(\sigma(t),k)$ is an ellipsoid, and therefore by definition the one sub-level sets of the function $\phi^{[\gamma]}(\cdot,\sigma(t))$ in $q_\text{perf}$ form the interior of $\R^{[\gamma]}(\sigma(t),k)$. That is, $\inter{\R}^{[\gamma]}(\sigma(t),k) = \{x\in\X \mid \phi^{[\gamma]}(x,\sigma(t)) <1 \}$.
\begin{figure}[t]
    \centering
    \scalebox{.5}{\input{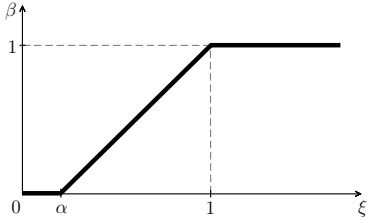}}
    \caption{$\beta_\alpha$ as a function of its argument for a given design parameter $\alpha \in [0,1)$.}
    \label{F:beta}
\end{figure}

Notice that $\phi^{[\gamma]}(x(t),\sigma(t))$ determines how ``deep'' inside the domain of $q_\text{perf}$ the trajectory is at time $t$ by evaluating the implicit surface function of the ellipsoid $\R^{[\gamma]}(\sigma(t),k) = \cl \Dom(q_\text{perf},\gamma,\sigma(t))$ at $x(t)$. While any distance measure can be employed, we use the readily-available definition of ellipsoid to simplify the analysis and implementation. The closer the trajectory is to the boundary of $\R^{[\gamma]}(\sigma(t),k)$ the greater the value of $\beta_\alpha(\phi^{[\gamma]}(x(t),\sigma(t)))$ and therefore the more pronounced the safety component $u_\text{safe}(t)$ of the input will be. On the other hand, if the trajectory is deep inside the domain of $q_\text{perf}$, $\phi^{[\gamma]}(x(t),\sigma(t))$ tends to $\alpha$ or less, and therefore $\beta_\alpha(\phi^{[\gamma]}(x(t),\sigma(t)))$ goes to zero. As a result, a greater emphasis is given to the performance component $u_\text{perf}(t)$ of the input. As such, the parameter $\alpha$ determines where within the domain of $q_\text{perf}$ the safety component $u_\text{safe}(t)$ should begin to kick in. Clearly, larger values of $\alpha$ imply more emphasis on performance, while smaller values yield a smoother transition between $q_\text{perf}$ and $q_\text{safe}$.



\begin{prop}\label{P:continuous_u}
    The modified control law \eqref{E:smooth1} is continuous across the automaton's transition on $(q_\text{\upshape perf},q_\text{\upshape safe})$.
\end{prop}

\begin{proof}
    For a given $\gamma$ and $k$, let $\hat{t}$ be the global time instant at which a transition to $q_\text{safe}$ occurs so that $x(\hat{t}) \in \partial\R^{[\gamma]}(\sigma(\hat{t}),k)$. We have that $\lim_{t\to \hat{t}} \beta_\alpha(\phi^{[\gamma]}(x(t),\sigma(t)))= \lim_{t\to \hat{t}}\phi^{[\gamma]}(x(t),\sigma(t)) =1$. Therefore,
    $\lim_{t \nearrow \hat{t}} u(t) =  (1-1) u_{\text{perf}}(\hat{t}) + u_{\text{safe}}(\hat{t})  = u_{\text{safe}}(\hat{t}) = \lim_{t \searrow \hat{t}} u(t)$.
\end{proof}

The same result holds for transitions on $(q_\text{safe},q_\text{perf})$ if $R(q_\text{safe},q_\text{perf},\gamma) = \gamma$. Such a policy will ensure a gradual change of the control law from one form to the other, resulting in less switching frequency between performance and safety (see \cite[Section~6.3]{kaynama_thesis} for further detail).  If the disturbance plays optimally (or there is no disturbance), we have two scenarios: (i) Exact arithmetic: Neither the modified nor the unmodified control policy causes chattering in the automaton. With the unmodified policy, as soon as the trajectory reaches the switching surface, the optimal control law in $q_\text{safe}$ keeps the vector field tangential to the surface via \eqref{E:dist_optimal_disturbance} (until it enters the interior of an ellipsoid corresponding to a different terminal direction). With the modified policy \eqref{E:smooth1}, the trajectory approaches the switching surface asymptotically; therefore, the automaton never has to transition to $q_\text{safe}$. (ii) Inexact arithmetic (e.g.\ due to floating point or truncation error in a simulation): The unmodified policy will likely cause the automaton as well as the control signal to chatter since small arithmetic errors can force the state back and forth across the switching surface.  With the modified policy the automaton is far less likely to chatter since the trajectory is only asymptotically approaching the switching surface and therefore the arithmetic errors necessary to cause it to jump across the switching surface would have to be larger. Furthermore, the control signal $u(\cdot)$ in this case is continuous across the switching surface and therefore the only discontinuities it can have would be due to discontinuities in $u_\text{perf}(\cdot)$ or changes in $\gamma$.

If the disturbance is allowed to play non-optimally (in the sense of safety), then it can choose to switch back and forth between driving the trajectory inward and outward with respect to the current ellipsoid.  In this case the unmodified control policy can be tricked into chattering (both in mode and control signal) as fast as the disturbance itself is willing to switch. In contrast, the modified policy will still generate a trajectory which is at worst asymptotically approaching the switching surface, and the non-optimal segments of the disturbance signal will only serve to keep the trajectory further from the switching surface via \eqref{E:dist_nonoptimal_disturbance}. Consequently, the conclusions of the previous paragraph for the modified policy also apply in this case.

\subsection{Infinite-Horizon Control}



Suppose that the initial time $s$ in the definition of $\sigma$ in \eqref{E:pseudotime_defn} is allowed to vary discontinuously. In such a case, $\sigma(\cdot)$ is only \emph{piecewise} continuous and non-decreasing. Assume that the controlled-invariance condition \eqref{E:control_inv_cond} in Proposition~\ref{P:control_inv_cond} holds for some $k\in \{1,\dots,\abs{P}\}$ in at least one terminal direction $\ell_\tau \in \M$. Denote the subset of all such terminal directions as $\M_\infty \subseteq \M$. Now, consider the modified hybrid automaton $\widetilde H$ depicted in Fig.\ ~\ref{F:automatonII}.
\begin{figure*}[t]
    \centering
    \scalebox{.63}{\input{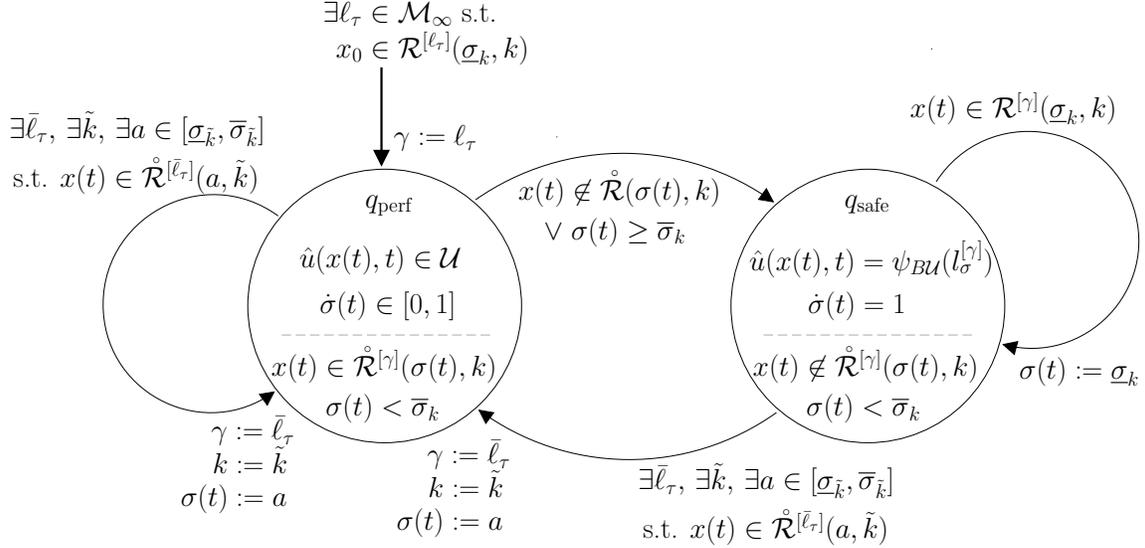}}
    \caption{The graph of the hybrid automaton $\widetilde H$ representing the infinite-horizon safety-preserving controller in Proposition~\ref{P:infinite_horizon_control}.}
    \label{F:automatonII}
\end{figure*}
We claim that $\widetilde H$ can be used to synthesize infinite-horizon safety-preserving control laws.

Before we prove this claim, let us describe a few components of $\widetilde H$: The initial state is now $\Init = (q_\text{perf}, x_0, 0, \{\gamma \mid \exists \ell_\tau \in \M_\infty,\, x_0 \in \R^{[\ell_\tau]}(\underline{\sigma}_k,k),\, \gamma=\ell_\tau \})$ where $k$ for a given $\ell_\tau$ corresponds to the first sub-interval for which the condition \eqref{E:control_inv_cond} has been satisfied. The edge $(q_\text{safe}, q_\text{safe})$ resets the pseudo-time downward via $(q_\text{safe}, q_\text{safe}, \sigma(t))\mapsto \underline{\sigma}_k$ and is enabled whenever the trajectory reaches the closure of the set $\R^{[\gamma]}(\underline{\sigma}_k,k)$ for a given $\gamma$ regardless of the value of $\sigma(t)$. The guard on $(q_\text{perf},q_\text{safe})$ is altered so as to enable a transition to $q_\text{safe}$ when $x(t)\not\in \inter{\R}^{[\gamma]}(\sigma(t),k)$ or $\sigma(t)\geq \overline{\sigma}_k$, or both. As we shall see shortly, when $\sigma(t)\geq \overline{\sigma}_k$ the automaton \emph{will} transition on $(q_\text{safe}, q_\text{safe})$ due to Proposition~\ref{P:control_inv_cond}; thus, $\widetilde{H}$ is non-blocking with respect to this condition. Finally, the guards on $(q_\text{perf},q_\text{perf})$ and $(q_\text{safe},q_\text{perf})$ are also modified so as to additionally permit transitions between any two distinct reachable tube segments corresponding to different directions in $\M_\infty$, or to allow jumps between the reachable sets within the same reachable tube segment.

With the above formulation we no longer require access to all intermediate maximal reachable tubes that would have been computed offline during the analysis phase. Instead, for every $\ell_\tau \in \M_\infty$ we only need the reachable tube over the sub-interval $[t_{k-1},t_k]=:[\underline{\sigma}_k,\overline{\sigma}_k]$ with $k$ being the first integer for which the invariance condition holds.

\begin{prop}[Infinite-Horizon Control]\label{P:infinite_horizon_control}
    Suppose the controlled-invariance condition in Proposition~\ref{P:control_inv_cond} holds. For any disturbance $v(\cdot)\in \Vl^\NA_{\Real^+}$, the feedback strategy generated by $\widetilde H$ can keep the trajectory $x(\cdot)$ of the system \eqref{E:linear_ss} with $x(0)=x_0$ contained in $\K$ for all time $t \geq 0$.
\end{prop}
\begin{proof}
    Fix $\gamma=\ell_\tau$ and $k$ according to $\Init$. It follows from Theorem~\ref{Thm:safety_preserving_control} and Proposition~\ref{P:control_inv_cond} that the control law generated by $\widetilde{H}$ is safety-preserving over $[\underline{\sigma}_k, \overline{\sigma}_k]$.
    To prove $\widetilde H$ generates an infinite-horizon control strategy, it is sufficient to show that for any $x_0 \in \R^{[\gamma]}(\underline{\sigma}_k,k)$, the pseudo-time $\sigma(t)$ cannot exceed $\overline{\sigma}_k$ and is repeatedly reset to $\underline{\sigma}_k$. We proceed by arriving at a contradiction: Assume that there does not exist such a strategy. This implies that the automaton can never transition on $(q_\text{safe},q_\text{safe})$. Now, begin in $q_\text{perf}$ using $\Init$. The automaton remains in $q_\text{perf}$ as long as $\sigma(t) < \overline{\sigma}_k$ and $x(t) \in \inter{\R}^{[\gamma]}(\sigma(t),k)$ during which time we have $\dot{\sigma}(t) \in [0,1]$. Thus, the trajectory $x(\cdot)$ will either eventually escape the interior of $\R^{[\gamma]}(\sigma(t),k)$ for all $\gamma$ for some $\sigma(t)\in [\underline{\sigma}_k,\overline{\sigma}_k)$, or $\sigma(t)=\overline{\sigma}_k$, or both. In all cases the automaton transitions from $q_\text{perf}$ to $q_\text{safe}$ where $\dot{\sigma}(t)=1$. At this point we have three scenarios: (i) If $\sigma(t)<\overline{\sigma}_k$ and $x(t) \in (\inter{\R}^{[\gamma]}(\sigma(t),k))^c$ the optimal control law in $q_\text{safe}$ steers the system towards $\R^{[\gamma]}(\overline{\sigma}_k,k)$. As such, $\sigma(t)$ advances until $\sigma(t)=\overline{\sigma}_k$ and $x(t)\in \R^{[\gamma]}(\overline{\sigma}_k,k)$. (ii) If $\sigma(t) = \overline{\sigma}_k$ and $x(t) \in \inter{\R}^{[\gamma]}(\sigma(t),k)$, then trivially $x(t) \in \inter{\R}^{[\gamma]}(\overline{\sigma}_k,k) \subset \R^{[\gamma]}(\overline{\sigma}_k,k)$. (iii) If $\sigma(t) = \overline{\sigma}_k$ and $x(t) \in (\inter{\R}^{[\gamma]}(\sigma(t),k))^c$ then it must be that $x(t) \in \partial \R^{[\gamma]}(\overline{\sigma}_k,k) \subset \R^{[\gamma]}(\overline{\sigma}_k,k)$. In all three scenarios the domain of $q_\text{safe}$ is no longer satisfied. But since via the invariance condition \eqref{E:control_inv_cond} we have $\R^{[\gamma]}(\overline{\sigma}_k,k) = K_k^{[\gamma]}(P) \subseteq \R^{[\gamma]}(\underline{\sigma}_k,k) = \Reach^{\sharp[\gamma]}_{\overline{\sigma}_k-\underline{\sigma}_k}(K_k^{[\gamma]}(P),\U,\V)$  (recall the shorthand notation \eqref{E:shorthand_notation_R}), then $x(t) \in \R^{[\gamma]}(\underline{\sigma}_k,k)$ and therefore the guard on $(q_\text{safe},q_\text{safe})$ is enabled, allowing $\widetilde H$ to transition on this edge and $\sigma(t)$ is reset to $\underline{\sigma}_k$.
\end{proof}

\paragraph{Example}
Consider a trivial planar system
\begin{equation}
    \dot{x} = \begin{bmatrix}
        \p0 & 2\\
        -2 & 0
    \end{bmatrix} x +
    \begin{bmatrix}
        1\\
        0.5
    \end{bmatrix} u +
    \begin{bmatrix}
        1\\
        1
    \end{bmatrix} v.
\end{equation}
Suppose $\K:=\E\left(\left[\begin{smallmatrix}
  0\\
  0
\end{smallmatrix}\right], \, \left[\begin{smallmatrix}
   0.25 &  0\\
     0  &  4
\end{smallmatrix}\right] \right)$, $\U:=[-1,1]$, and $\V:=[-0.1,0.1]$. For simplicity of demonstration, we use a single terminal direction $\ell_\tau = \tr{\begin{bmatrix} 1  &  1 \end{bmatrix}}$ in the offline approximation algorithm although the same analysis can be performed using multiple terminal directions. We only run the analysis algorithm over the time interval $[0,1]$ using a uniform partitioning with $\abs{P}=100$. The condition in Proposition~\ref{P:control_inv_cond} is satisfied at the end of the $94$-th sub-interval ($k=6$) at which point we terminate the algorithm. We then discard all computed intermediate maximal reachable tubes except for the tube over the sub-interval where the controlled-invariance condition holds, and proceed to execute the hybrid controller $\widetilde H$ using only this tube. For simulation purposes we consider an initial condition $x_0 = \tr{\begin{bmatrix} 0.3  &  -0.7 \end{bmatrix}}$ and assume that the disturbance (unknown to the controller) is a random signal with uniform distribution on $\V$, i.e.\ $v(t) \sim \operatorname{uniform}(-0.1,0.1)$. Since $u$ in $q_\text{perf}$ can be chosen arbitrarily in $\U$, we simply apply $u=\min(\U)=-1$ in this mode. The results are shown in Figs.\ ~\ref{F:infty} and ~\ref{F:infty_ctrl}. We confirm that (a) the sub-interval reachable tube is robust controlled-invariant, and (b) the controller is capable of preserving safety over an \emph{infinite} horizon (although we manually terminate the simulation at $t=25\, \text{s}$). The closed-loop trajectory converges to a nearly periodic orbit at about $t\approx 2.5\, \text{s}$ after which the controller enters $q_\text{safe}$ only occasionally so as to maintain safety and invariance. Chattering is present, but the modified ``smoothing'' policy discussed in Section~\ref{S:smooth} could be employed if such behavior is undesirable.
\begin{figure}[t]
  \centering
    \begin{minipage}[t]{0.48\linewidth}
      \centering
        \scalebox{.85}{%
        \begin{lpic}[l(3mm),r(0mm),t(0mm),b(3mm),draft,clean]{infty2(.55)}
            \lbl[t]{55,-1;$x_1$}
            \lbl[lb]{2.5,48,90;$x_2$}
            \lbl[lb]{65,36;$x_0$}
        \end{lpic}%
        }
        \caption{The infinite-horizon safety-preserving control of a planar system using $\widetilde H$. The simulation is manually terminated at $t=25\, \text{s}$.  The constraint set $\K$ (red/dark) and the sub-interval maximal reach tube (green/light) for a single terminal direction are shown. The synthesized control law renders the tube robust controlled-invariant over an infinite horizon.}
        \label{F:infty}
    \end{minipage}\hfill
    \begin{minipage}[t]{0.48\linewidth}
      \centering
        \scalebox{.85}{%
        \begin{lpic}[l(3mm),r(0mm),t(0mm),b(3mm),draft,clean]{infty2_ctrl(.6)}
            \lbl[lb]{2.5,23,90;$u$}
            \lbl[lb]{65,-2;$t$}
        \end{lpic}%
        }
        \caption{Feedback control law corresponding to Fig.~\ref{F:infty}. The high switching frequency, particularly at the beginning portion of the signal, is due to chattering (which can be avoided using the modified policy in Section~\ref{S:smooth}, if desired).}
        \label{F:infty_ctrl}
    \end{minipage}
\end{figure}

\section{Example Application}\label{S:Application}

Consider the twelve-dimensional linearized model of an agile quadrotor unmanned aerial vehicle (UAV) in \cite{Cowling_quadrotor_paper}:
\begin{equation}
    \dot x = Ax + Bu + Gv.
\end{equation}
The state
\begin{equation}
    x = \tr{\begin{bmatrix}
    \mathrm{x} & \mathrm{y} & \mathrm{z} & \dot{\mathrm{x}} & \dot{\mathrm y} & \dot{\mathrm z} & \phi & \theta & \psi & \dot{\phi} & \dot{\theta} & \dot{\psi}
    \end{bmatrix}}
\end{equation}
is comprised of translational position of the UAV in $[\text{m}]$ with respect to a global origin, their derivatives (velocities in $\mathrm x$, $\mathrm y$, $\mathrm z$ directions) in $[\text{m}/\text{s}]$, the Eulerian angles roll ($\phi$), pitch ($\theta$), and yaw ($\psi$) in $[\text{rad}]$, and their respective derivatives (angular velocities) in $[\text{rad}/\text{s}]$. The control vector $u = \tr{\begin{bmatrix}u_1 & u_2 & u_3 & u_4 \end{bmatrix}}$ consists of the total normalized thrust in $[\text{m}/{\text{s}^2}]$ and second-order derivatives $\ddot{\phi}$, $\ddot{\theta}$, $\ddot{\psi}$ of the Eulerian angles in $[\text{rad}/{\text{s}^2}]$, respectively. Note that this system is under-actuated (there are six degrees-of-freedom but only four actuators). The system matrices $A$ and $B$ are obtained by linearizing the equations of motion of the quadrotor about a hover condition $\phi = 0$, $\theta = 0$, and $u_1 = g$ (with $g \approx 9.81$ being the gravitational constant); cf.\ \cite{Cowling_quadrotor_paper} for values of these matrices. The disturbance $v \in \Real$ (unknown to the system) is the wind being applied to the vehicle in the direction of $\dot{\mathrm{x}}$, $\dot{\mathrm y}$, and $\dot{\mathrm z}$ and is modeled as an i.i.d.\ random process uniformly distributed over the interval $[0,0.1]$ measured in $[\text{m}/\text{s}]$. Thus the matrix $G$ is a column vector of ones for the states $\dot{\mathrm{x}}$, $\dot{\mathrm y}$, and $\dot{\mathrm z}$ and zeros for all other states.

For safe operation of this UAV the Eulerian angles $\phi$ and $\theta$ and the speed $V$ are bounded such that $V := \norm{\begin{bmatrix} \dot{\mathrm{x}} & \dot{\mathrm y} & \dot{\mathrm z} \end{bmatrix}}_2 \leq 5$ and $\phi,\theta \in [-\frac{\pi}{2}, \frac{\pi}{2}]$. The angular velocities are assumed to be constrained to $\dot{\phi},\dot{\theta},\dot{\psi} \in [-3,3]$ (cf.\ \cite{Cowling_quadrotor_paper} for further details). We therefore approximate the flight envelope $\K$ by an ellipsoid $\K_{\varepsilon}$ of maximum volume inscribed within these bounds while further assuming that the location of the UAV is to remain within a ball of radius $3$ centered at $(0,0,4)$ in $\{\mathrm{x},\mathrm{y},\mathrm{z}\}$ coordinates (i.e. the vehicle must safely fly within the range of $1$ to $7 \, \text{m}$ above the ground in $\textrm{z}$ direction). The input vector $u$ is constrained by the hyper-rectangle $\U:=[0.5, 5.4] \times [-0.5, 0.5]^3$ which we also under-approximate by a maximum volume inscribed ellipsoid $\U_{\varepsilon}$. For the horizon $\tau=2 \,\text{s}$ we require $x(t) \in \K_{\varepsilon}$ and $u(t)\in \U_{\varepsilon}$ $\forall t\in [0,\tau]$ despite the disturbance.

Suppose that to operate the UAV a Linear Quadratic Regulator (LQR) is designed that satisfies the performance functional
\begin{equation}
    J(u) := \int_0^\infty (\tr{x}Qx + \tr{u}Ru) dt
\end{equation}
subject to the system dynamics with $Q = 10^{-5} I_{12}$ and $R = \operatorname{diag}(10^{-6}, 10^8, 10^8, 10^8)$ (similar to those reported in \cite{Cowling_quadrotor_paper}). The corresponding state-feedback control law $u_\text{lqr}$ that minimizes $J(u)$ is constructed such that the trajectory of the closed-loop system moves from $x_0$ towards the steady-state point $x_{ss}=\tr{\begin{bmatrix}0 & 0  & 5 & 0 & 0 & 0 & 0 & 0 & 0  & 0  & 0  & 0\end{bmatrix}}$. 

Since the control authority is constrained by the set $\U_{\varepsilon}$, applying the above LQR control law results in saturation such that $u=\operatorname{sat}(u_\text{lqr})$ is effectively applied to the system. Here, the saturation function $\operatorname{sat}\colon \Real^4 \to \U_{\varepsilon}$ is determined by the support vector of $\U_{\varepsilon}= \E(\mu,U)$ in the direction of $u_\text{lqr}/\norm{u_\text{lqr}}$ and is defined as
\begin{equation}
    \operatorname{sat}(u_\text{lqr}) :=
    \begin{cases}
        u_\text{lqr} & \text{if \; $u_\text{lqr} \in \U_{\varepsilon}$};\\
        \mu + U u_\text{lqr} \Inner{u_\text{lqr},\, U u_\text{lqr}}^{-1/2} & \text{if \; $u_\text{lqr} \not\in \U_{\varepsilon}$}.
    \end{cases}
\end{equation}

Consider the initial condition
\begin{equation}
\begin{multlined}[t][12cm]
    x_0 = \tr{\left[\begin{matrix}-0.4032 & 0.7641 & 3.6437 & -1.2406\end{matrix} \right.\\
            \begin{matrix}0.0165 &  3.0335 &  -0.0789 &  -0.4835\end{matrix}\\
            \left. \begin{matrix}-0.3841 &   0.0375 &   0.6806 &   0.5509 \end{matrix}\right]}.
\end{multlined}
\end{equation}
Without engaging a safety controller, the (saturated) LQR results in violation of the safety constraint $\K_{\varepsilon}$ at $t=1.8\, \text{s}$; Figs.\ ~\ref{F:UAV_traj_nosafe} and ~\ref{F:UAV_control_nosafe}. In contrast, the hybrid controller presented in this paper is capable of preserving safety; Figs.\ ~\ref{F:UAV_traj_safe} and ~\ref{F:UAV_control_safe}. A piecewise ellipsoidal under-approximation of $\Disc_{[0,2]}(\K_{\varepsilon},\U_{\varepsilon})$ was precomputed in $35\, \text{min}$ using $15$ random terminal directions and a uniform time interval partitioning with $\abs{P}=200$ on an Intel-based machine with $2.8 \, \text{GHz}$ CPU and $3\,\text{GB}$ RAM running single-threaded $32$-bit \textsc{MATLAB} 7.5 and Ellipsoidal Toolbox v.1.1.3 \cite{KV06}.  The controller is configured so as to permit the same (saturated) LQR in $q_\text{perf}$, while switching to $q_\text{safe}$ when necessary. A smoothing parameter $\alpha=0.9$ is chosen in \eqref{E:Beta_intermsof_phi_alpha}, and the modified strategy \eqref{E:smooth1} is employed to prevent chattering of the control signal. Fig.\ ~\ref{F:UAV_u_percent} shows the percentage of the performance component of the input signal. The pseudo-time rate of $\dot{\sigma}=1$ was enforced in $q_\text{perf}$ so as to not sacrifice performance for prolonged safety. On the other hand, if extended safety is desired, freezing the pseudo-time in $q_\text{perf}$ results in maintaining safety for $4.43 \,\text{s}$ using the sets precomputed for only $\tau=2\, \text{s}$. See \cite[Section~6.6.3]{kaynama_thesis} for another example on safety-based control of anesthesia. We note that no current Eulerian method is capable of directly synthesizing a safety-preserving controller for such a high-dimensional system.

\begin{figure}[t]
  \centering
    \begin{minipage}[t]{0.48\linewidth}
      \centering
        \includegraphics[scale=0.6]{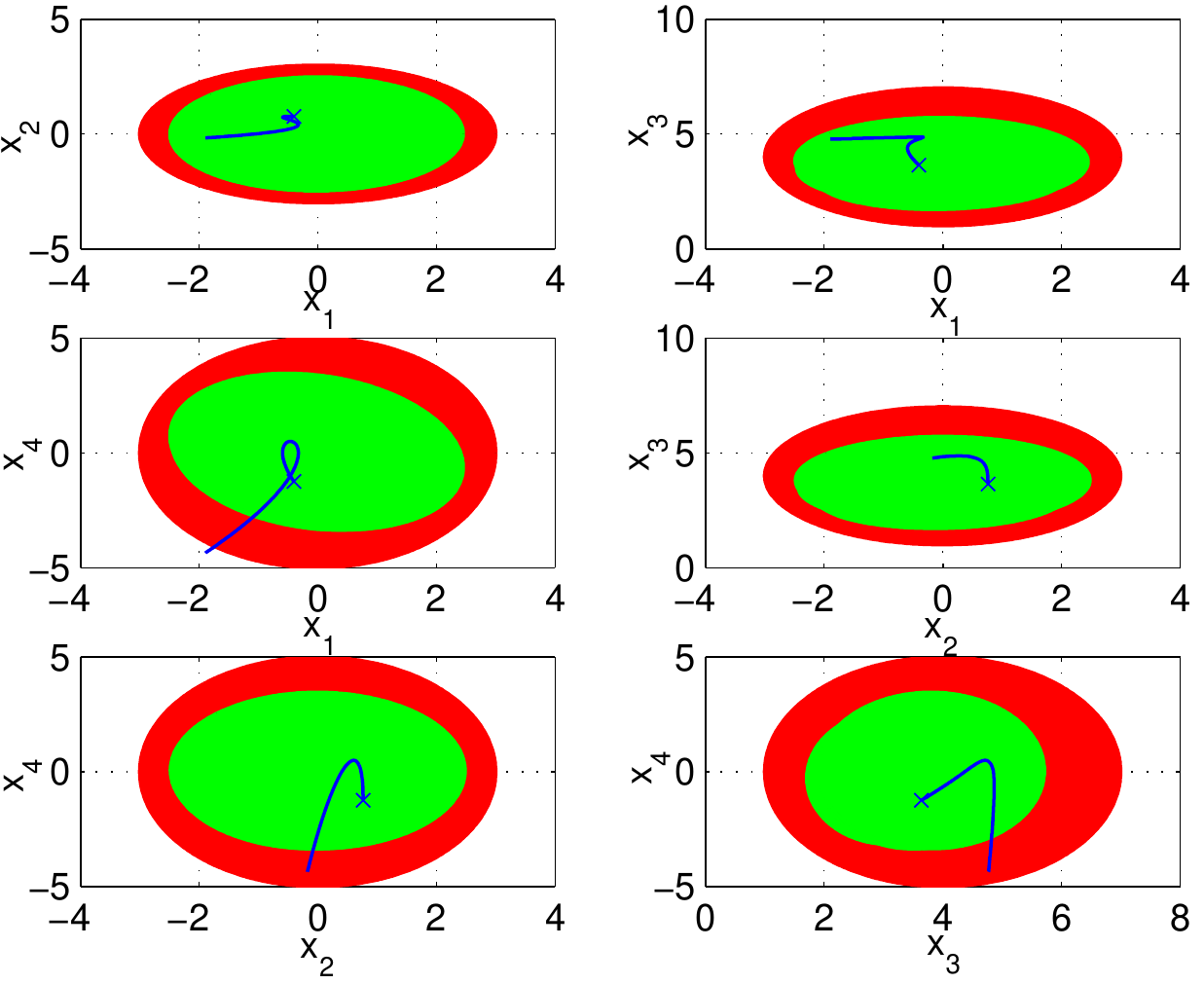}
        \caption{2D projections of the first four dimensions of the closed-loop UAV trajectories using saturated LQR with no safety control. $x_0$ is marked by `$\times$'. The constraint set $\K_{\varepsilon}$ (red/dark) and a piecewise ellipsoidal under-approximation of $\Disc_{[0,2]}(\K_{\varepsilon},\U_{\varepsilon})$ (green/light) are also shown. Safety is violated due to unaccounted for actuator saturations.}
        \label{F:UAV_traj_nosafe}
    \end{minipage}\hfill
    \begin{minipage}[t]{0.48\linewidth}
      \centering
        \includegraphics[scale=0.5]{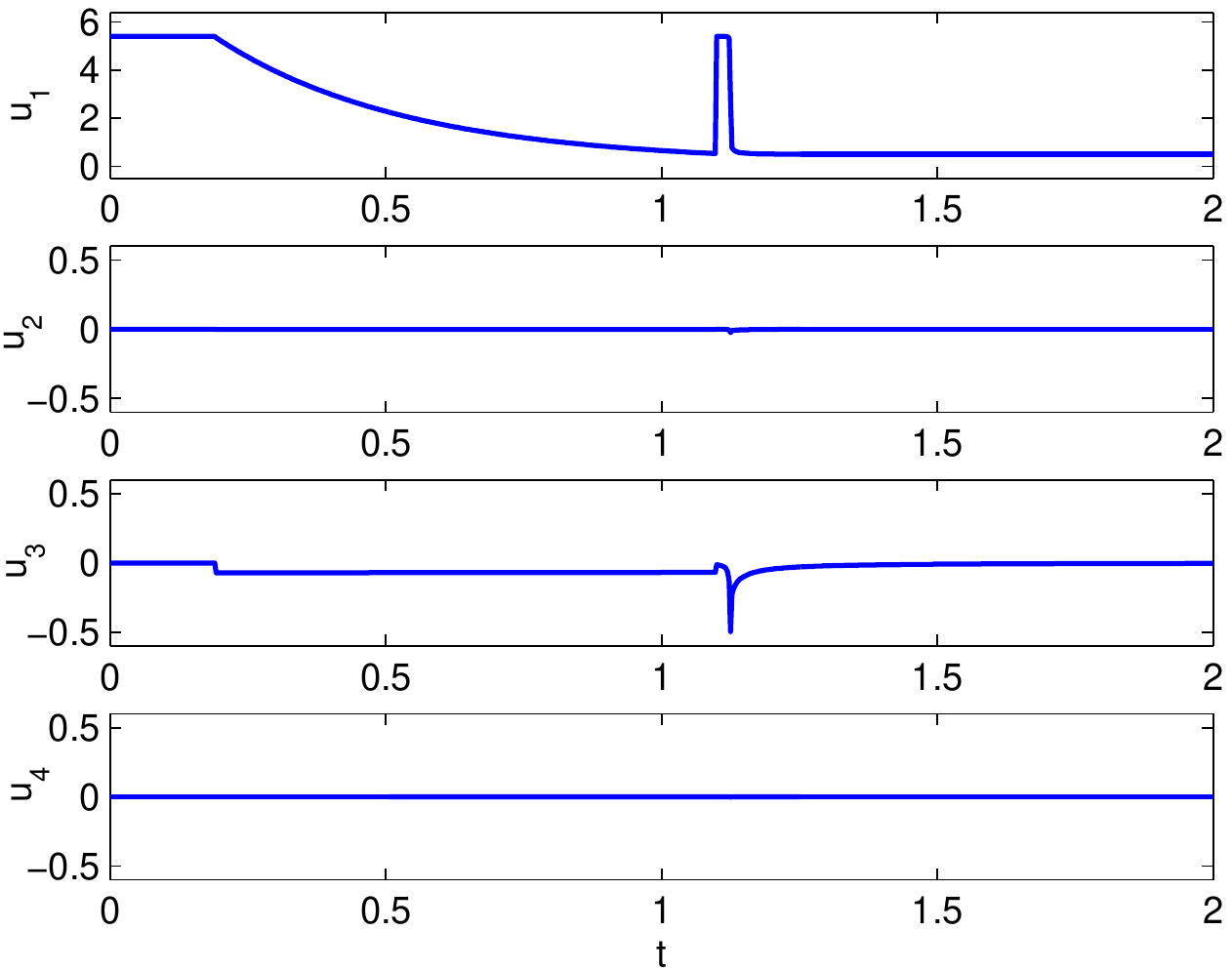}
        \caption{Control signals corresponding to Fig.~\ref{F:UAV_traj_nosafe} for the UAV (saturated LQR; no safety)}
        \label{F:UAV_control_nosafe}
    \end{minipage}
\end{figure}

\begin{figure}[t]
  \centering
    \begin{minipage}[t]{0.48\linewidth}
      \centering
        \includegraphics[scale=0.6]{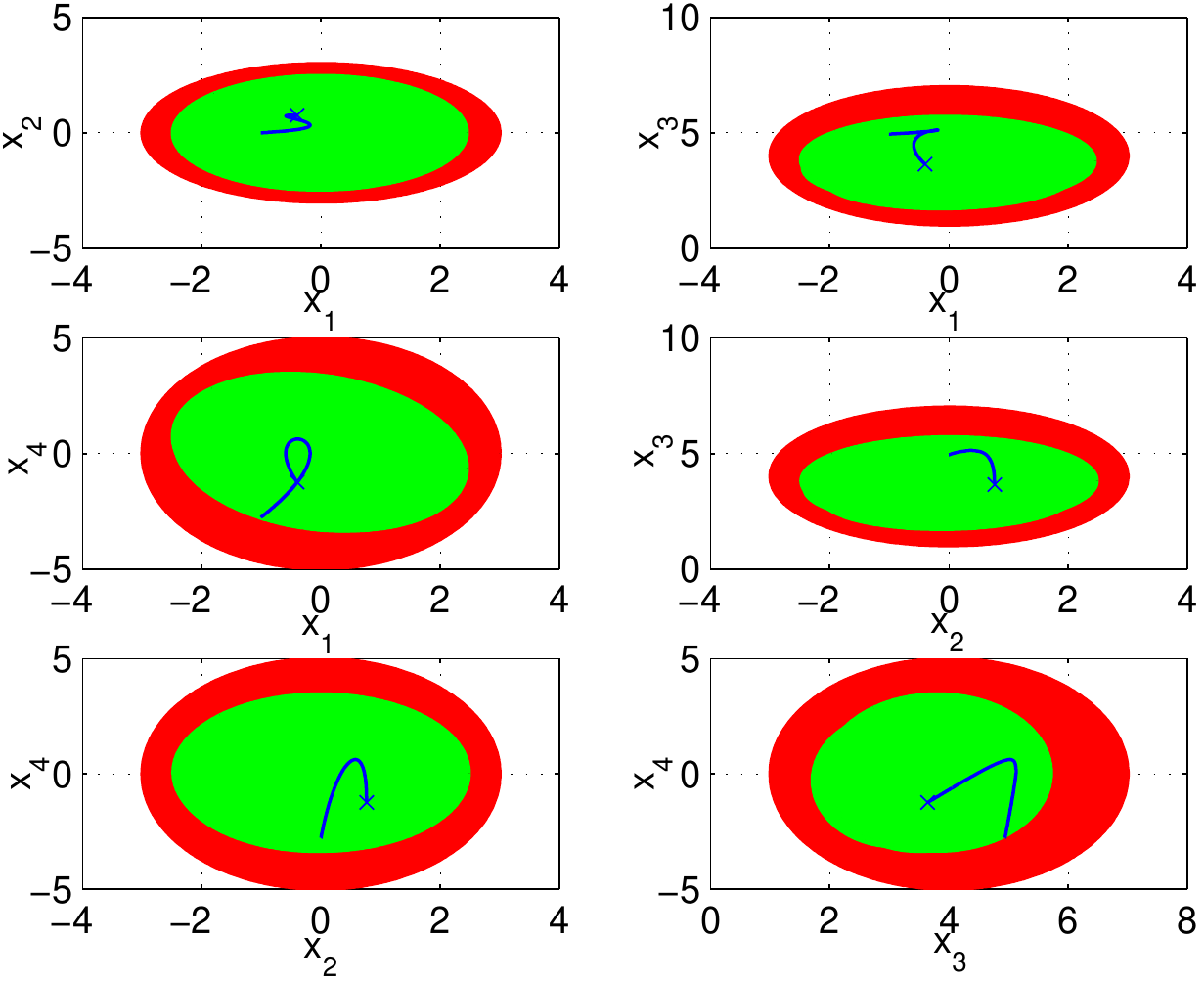}
        \caption{2D projections of the first four dimensions of the closed-loop UAV trajectories using safety-preserving control strategies. Safety is maintained over the desired horizon in spite of the actuator saturations in the LQR controller and the unknown disturbance.}
        \label{F:UAV_traj_safe}
    \end{minipage}\hfill
    \begin{minipage}[t]{0.48\linewidth}
      \centering
        \includegraphics[scale=0.5]{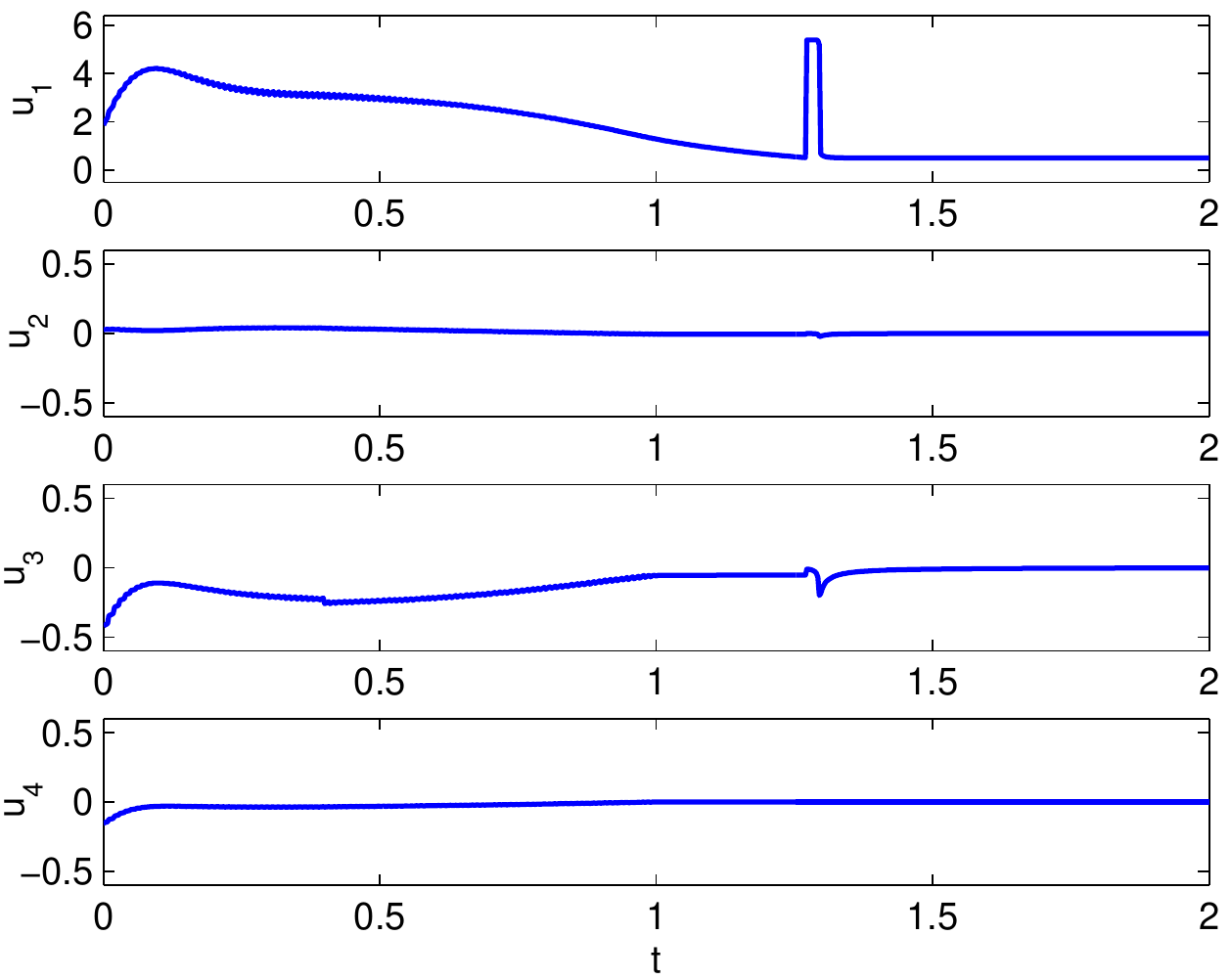}
        \caption{Safety-preserving control signals corresponding to Fig.~\ref{F:UAV_traj_safe} for the UAV. The modified policy \eqref{E:smooth1} prevents chattering to a great extent.}
        \label{F:UAV_control_safe}
    \end{minipage}
\end{figure}

\begin{figure}[t]
      \centering
        \includegraphics[scale=0.5]{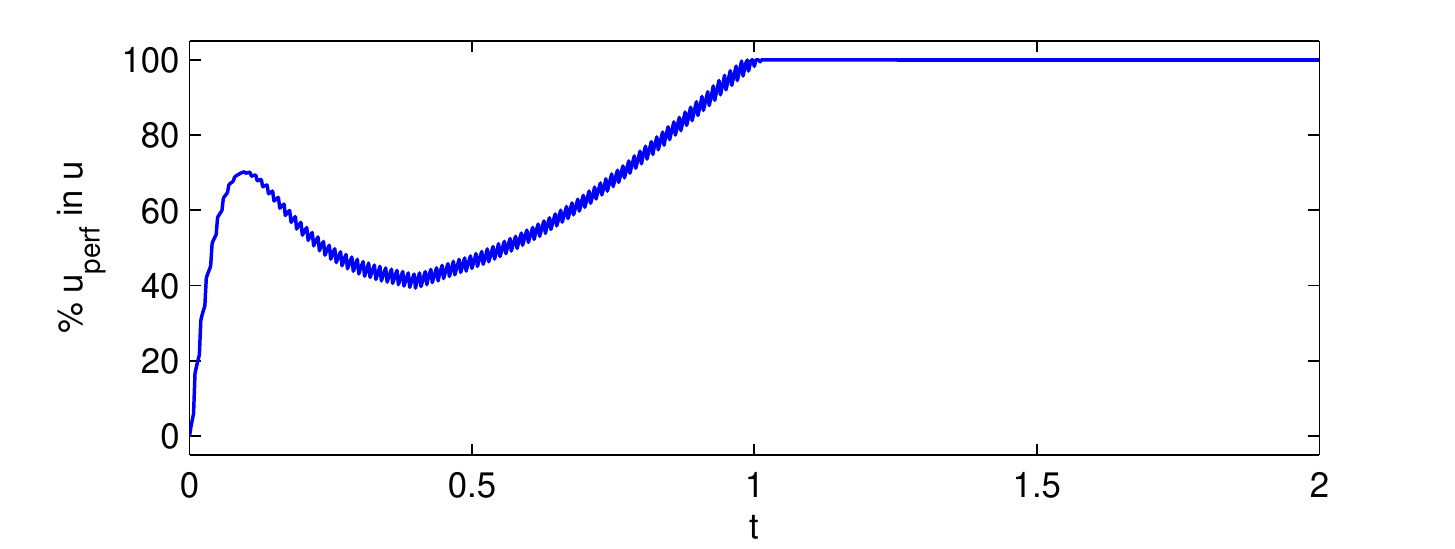}
        \caption{Percentage of $u_\text{perf}(t)$ in the vector $u(t)$ when using the modified policy \eqref{E:smooth1} with $\alpha=0.9$ for the UAV. The initial condition $x_0$ is close to the switching surface (unrecognizable from the 2D projection plots). A value of $0\%$ would correspond to $u(t)= u_\text{safe}(t)$, while a value of $100\%$ would correspond to when $x(t)$ is within the $\alpha$ sub-level set of the domain of $q_\text{perf}$ and thus $u(t)=u_\text{perf}(t)$. Increasing $\alpha$ would further emphasize the $u_\text{perf}$ component, at the cost of a more aggressive control policy.}
        \label{F:UAV_u_percent}
\end{figure}

\section{Conclusions \& Future Work}\label{S:conc}

We presented a scalable set-valued robust safety-preserving feedback controller, applicable to high-dimensional LTI systems. To this end, we extended our previous results on scalable approximation of the viability kernel to the case where a constrained continuous-time system is subject to bounded additive uncertainty/disturbance and therefore the discriminating kernel is to be computed. The controller is formulated as a hybrid automaton that incorporates a pseudo-time variable when prolonged (potentially infinite-horizon) safety using finite-horizon computations is desired. A heuristic method was presented that reduces chattering to a great extent and results in a control signal that is continuous across transitions of the automaton. The results were demonstrated on a twelve-dimensional model of a quadrotor where safety (flight envelope) was maintained despite actuator saturations and unknown disturbances.

Representing the discriminating kernel in terms of robust maximal reachable sets has the advantage of employing Lagrangian methods for its approximation. This paves the way for potentially more scalable computation of the kernel than with grid-based techniques even under nonlinear dynamics (should an appropriate Lagrangian tool be developed in the future). Since the presented safety-preserving controller in this paper is constructed based on ellipsoidal techniques and semidefinite programmes, its computational complexity in the offline phase---its most costly aspect---is cubic in the dimension of the state (and linear in the number of time partitions and the number of terminal directions). Thus, it can be employed for treatment of systems with an order of magnitude higher dimension that is currently possible with grid-based techniques. Extension of this work to the discrete-time case appears to be straightforward. We have investigated a sampled-data extension in \cite{kaynama_NAHS2012}. We are currently studying the advantages and disadvantages of our scalable technique as compared to more conventional approaches such as MPC or Hamilton-Jacobi formulations.



\appendix \label{appendix}

\section{Proofs of Propositions~\ref{P:continuous_subset_Disc} and \ref{P:precision_Disc}}\label{App}

The proofs presented here are straightforward extensions of our previous results in \cite{kaynama_HSCC2012}.

\begin{proof}[Proof (Proposition~\ref{P:continuous_subset_Disc})]
    Fix a partition $P$ of $[0,\tau]$ and take a point $x_0 \in K_0(P)$. By construction of $K_0(P)$, this means that for each $k \in \{ 1, \ldots, \abs{P}\}$ there is some point $x_k \in K_k(P)$ and, for every disturbance $v_k \in \Vl^\NA_{[0, t_k-t_{k-1}]}$, a control $u_k \in \Ul^\CL_{[0, t_k-t_{k-1}]}$ such that $x_k$ can be reached from $x_{k-1}$ at time $t=t_k-t_{k-1}$ using $u_k$. Thus, taking the concatenation of the inputs $u_k$ and $v_k$, we get for every disturbance $v \in \Vl^\NA_{[0,\tau]}$ a control $u \in \Ul^\CL_{[0,\tau]}$ such that the solution $x\colon[0,\tau] \to \X$ to the initial value problem $\dot{x}=f(x,u,v)$, $x(0)=x_0$, satisfies $x(t_k) = x_k \in K_k(P) \subseteq \{x \in \K \mid \dist(x, \K^c) \geq M \norm{P}\} $. We claim that this guarantees that $x(t) \in \K$ $\forall t \in [0,\tau]$. Indeed, any $t \in [0,\tau)$ lies in some interval $[t_k,t_{k+1})$. Since $f$ is bounded by $M$, we have
    \begin{equation}\label{E:max_dist_travel}
    \begin{split}
      \norm{x(t)-x(t_k)} &\leq \int_{t_k}^{t} \norm{\dot{x}(s)} ds \leq M(t-t_k) < M(t_{k+1} - t_k) \leq M \norm{P}.
    \end{split}
    \end{equation}
    Further, $x(t_k) \in K_k(P)$ implies $\dist(x(t_k),\K^c) \geq M\norm{P}$. Combining these, we see that
    \begin{equation}
        \begin{split}
            \dist(x(t),\K^c) &\geq \dist(x(t_k),\K^c) -  \norm{x(t)-x(t_k)} > M \norm{P} -M \norm{P} = 0
        \end{split}
    \end{equation}
    and hence $x(t) \in \K$. Thus, $x_0 \in \Disc_{[0,\tau]}(\K,\U,\V)$. 
\end{proof}

\begin{proof}[Proof (Proposition~\ref{P:precision_Disc})]
    The second inclusion in \eqref{E:precision_Disc} follows directly from Proposition~\ref{P:continuous_subset_Disc}. To prove the first inclusion, take a state $x_0 \in \Disc_{[0,\tau]}{(\inter{\K},\U,\V)}$. For every disturbance $v\in \Vl^\NA_{[0,\tau]}$ there exists a control $u \in \Ul^\CL_{[0,\tau]}$ such that the trajectory $x(\cdot)$ satisfies $x(t)\in \inter{\K}$ $\forall t \in [0,\tau]$. Since $\inter{\K}$ is open, $\forall x \in \inter{\K}$ \, $\dist(x,\K^c) > 0$. Further, $x \colon [0,\tau] \to \X$ is continuous so the function $t \mapsto \dist(x(t), \K^c)$ is continuous on the compact set $[0,\tau]$. Thus, we can define $d>0$ to be its minimum value. Take a partition $P$ of $[0,\tau]$ such that $\norm{P} < d/M$. We need to show that $x_0 \in K_0(P)$.

    First note that our partition $P$ is chosen such that $\dist(x(t),\K^c) > M \norm{P}$ $\forall t \in [0,\tau]$. Hence $x(t_k) \in K_{\abs{P}} (P)$ for all $k = 0, \ldots, \abs{P}$. To show that $x(t_{k-1}) \in \Reach^\sharp_{t_{k} - t_{k-1}} (K_k (P),\U,\V)$ for all $k = 1, \ldots,\abs{P}$, consider the tokenizations $\{u_k\}$ and $\{v_k\}$ of the control $u$ and the disturbance $v$, respectively, corresponding to $P$. It is easy to verify that for all $k$ we can reach $x(t_k)$ from $x(t_{k-1})$ at time $t_k-t_{k-1}$ for every $v_k$ using the control input $u_k$. Thus, in particular, we have $x_0 = x(t_0) \in \Reach^\sharp_{t_{1} - t_{0}} (K_1 (P),\U,\V)$. So $x_0 \in K_0(P)$. Hence $\Disc_{[0,\tau]}(\inter{\K},\U,\V) \subseteq \bigcup_{P \in \mathscr{P}([0,\tau])} K_0(P)$.
\end{proof}

\section*{Acknowledgment}

The first author thanks Dr.\ Alex A.\ Kurzhanskiy for helpful discussions on the synthesis of maximal reachability control laws using Ellipsoidal Toolbox.

\bibliographystyle{IEEEtran}
\bibliography{IEEEfull,viabET_v2,anesthesia_related}

\end{document}